\documentclass[11pt,a4paper,subeqn,reqno]{article}

\usepackage[hyperpageref]{backref}

\usepackage{amsmath,amsfonts,amssymb,mathtools,amsthm}
\usepackage{bbm} 
\usepackage{epsfig,mathrsfs,enumerate,graphics}
\usepackage[usenames,dvipsnames]{color} 
\usepackage{dsfont,graphicx,adjustbox,upgreek,wrapfig,caption,subcaption,tikz}
\usepackage{epstopdf,latexsym,colortbl,shuffle}

\usepackage[all]{xy} \entrymodifiers={!!<0pt,0.7ex>+}

\usepackage{url}

\usepackage[hyperindex=true, 
					hidelinks]{hyperref}

\newtheorem{theorem}{Theorem}

\newtheorem{lemma}{Lemma}[section]

\newtheorem{proposition}[lemma]{Proposition}
\newtheorem{claim}[lemma]{Claim}

\theoremstyle{definition}
\newtheorem{notation}[lemma]{Notation}
\newtheorem{example}[lemma]{Example}
\newtheorem*{examples}{Examples}
\newtheorem{definition}[lemma]{Definition}
\newtheorem{definition-lemma}[lemma]{Definition-Lemma}
\newtheorem{definition-theorem}[lemma]{Definition-Theorem}

\newtheorem{remark}[lemma]{Remark}

\numberwithin{equation}{section}

\newcounter{parag}[subsection]

\newcounter{parage}[section]
\renewcommand{\theparage}{{\bf\thesection.\arabic{parage}}}
\newcommand{\parage}{\medskip \addtocounter{parage}{1} 
\noindent{\theparage\ } }

\newcounter{paraga}

\newcommand{\be}{\beta}

\newcommand{\Ga}{{\Gamma}}
\newcommand{\de}{\delta}

\newcommand{\om}{\omega}

\newcommand{\Om}{\Omega}

\newcommand{\sig}{{\sigma}}

\newcommand{\tht}{{\theta}}

\newcommand{\ph}{\varphi}

\newcommand{\ze}{{\zeta}}

\newcommand{\bino}[2]{\mbox{$
\begin{pmatrix}#1\\#2\end{pmatrix}%
                $}}

\newcommand{\ov}{\overline}

\newcommand{\ID}{\mathop{\hbox{{\rm Id}}}\nolimits}

\newcommand{\I}{{\mathrm i}}
\newcommand{\dd}{{\mathrm d}}

\newcommand{\ee}{\mathrm e}

\newcommand{\pa}{\partial}

\newcommand{\ii}{^{-1}}

\newcommand{\ie}{{\emph{i.e.}}\ }

\newcommand{\eg}{{\it e.g.}\ }

\newcommand{\wrt}{{with respect to}}

\newcommand{\C}{\mathbb{C}}

\newcommand{\R}{\mathbb{R}}

\newcommand{\cB}{\mathcal{B}}

\newcommand{\cQ}{\mathcal{Q}}

\DeclarePairedDelimiter\abs{\lvert}{\rvert}%
\newcommand{\defeq}{\coloneqq} 

\newcommand{\col}{\colon\thinspace}          

\newcommand{\gS}{\mathscr S}       

\newcommand{\eith}{\ee^{\I\tht}}

\newcommand{\begla}{\begin{equation}}
\newcommand{\beglab}[1]{\begin{equation}	\label{#1}}
\newcommand{\edla}{\end{equation}}

\newcommand{\imp}{\ens\Longrightarrow\ens}
\newcommand{\Imp}{\quad\Longrightarrow\quad}

\newcommand{\Li}{\operatorname{Li}}

\newcommand{\ot}{{\otimes}}

\newcommand{\tin}{\tilde}
\newcommand{\ti}{\widetilde}

\newcommand{\wh}{\widehat}

\newcommand{\ens}{\enspace}

\newcommand{\footremember}[2]{%
\footnote{#2}
\newcounter{#1}
\setcounter{#1}{\value{footnote}}%
}
\newcommand{\footrecall}[1]{%
\footnotemark[\value{#1}]%
}

\title{On the Moyal Star Product of Resurgent Series}

\author{%
  Y.~Li\footnote{Chern Institute of Mathematics 
  and LPMC, Nankai university, Tianjin 300071, China}
  \and
  D.~Sauzin\footnote{CNRS -- 
    Observatoire de Paris, PSL Research University, 75014 Paris, France}
  \footremember{CNU}{Department of Mathematics, Capital
    Normal University, Beijing 100048, China}
  \and
  S.~Sun\footrecall{CNU} \footnote{Academy for Multidisciplinary
    Studies, Capital Normal University, Beijing 100048, China}%
}

\date{December 2020}

\begin{document}

\thispagestyle{empty}

\maketitle

\vspace{.7cm}


\begin{abstract}
  We analyze the Moyal star product in deformation quantization from
  the resurgence theory perspective.  By putting algebraic conditions
  on Borel transforms, one can define the space of ``algebro-resurgent
  series'' (a subspace of $1$-Gevrey formal series in $i\hbar/2$ with
  coefficients in $\C\{q,p\}$),
  which  we show is stable under Moyal star product.

\end{abstract}




\bigskip

\section{Introduction}   \label{secIntro}

\parage
The Moyal star product of two formal series~$\tin f$ and~$\tin g$
in $\C[[t,q_1,\ldots,q_N,p_1,\ldots,p_N]]$ is defined by the formula
\beglab{eqdefstarM}
\tin f\star_M \tin g \defeq \tin f\,\tin g+\sum_{k\geq
  1} \, \frac{t^k}{2^k k!} \,
%
%
\sum_{n=0}^k(-1)^{k-n} \bino{k}{n}
\Big(\pa_p^n \, \pa_q^{k-n} \tin f\Big) \Big(\pa_p^{k-n} \, \pa_q^{n} \tin g\Big)
%
%
%
\edla
when $N=1$, and by an analogous formula when the number of degrees of
freedom~$N$ is larger than~$1$.
We get a non-commutative associative algebra
$\big(\C[[t,q_1,\ldots,q_N,p_1,\ldots,p_N]],\star_M \big)$,
the unit of which is the constant series~$1$;
the corresponding commutator induces a
deformation of the Poisson bracket
$\{ .\,,.\}$
associated with the standard
symplectic structure $\dd p_1\wedge \dd q_1 + \cdots \dd p_N\wedge \dd
q_N$, inasmuch as
\begla
\tin f \star_M \tin g - \tin g \star_M \tin f = t \big\{ \tin f,
\tin g \big\} + O(t^2).
\edla
The main result of this paper is
%
\begin{theorem}   \label{thmMoyalN}
If $\tin f$ and~$\tin g$ are algebro-resurgent series, then so is their
Moyal star product $\tin f\star_M \tin g$.
\end{theorem}

Here, ``algebro-resurgence'' is a property of a formal series
$\tin f(t,q_1,\ldots,q_N,p_1,\ldots,p_N)$ defined in terms of its
formal Borel transform $\hat f(\xi,q_1,\ldots,q_N,p_1,\ldots,p_N)$ \wrt~$t$,
which is required to be convergent, \ie
$\hat f(\xi,q_1,\ldots,q_N,p_1,\ldots,p_N) \in
\C\{\xi,q_1,\ldots,q_N,p_1,\ldots,p_N\}$, and to admit analytic
continuation along all the paths which start close enough to the origin
of~$\C^{1+2N}$ and avoid a certain proper algebraic subvariety (which
depends on~$\tin f$). More details will be given in due time.

This variant of \'Ecalle's definition of resurgence \cite{E84},
\cite{Eca85} was introduced by M.~Garay, A.~de Goursac and D.~van Straten in
\cite{GGS}, where they state the
%
\begin{theorem}   \label{thmStdN}
If $\tin f$ and~$\tin g$ are algebro-resurgent series, then so is their
standard star product
\beglab{eqdefstarS}
\tin f\star_S \tin g \defeq
%
%
\sum_{k_1,\cdots,k_N\ge0}
  \frac{t^{k_1+\cdots+k_N}}{k_1!\cdots k_N!} \Big(\partial_{p_1}^{k_1} \cdots
  \partial_{p_N}^{k_N} \tin f\Big)\Big( \partial_{q_1}^{k_1} \cdots
  \partial_{q_N}^{k_N} \tin g\Big).
\edla
\end{theorem}

However, their proof of Theorem~\ref{thmStdN} is not valid, due to a
flaw in one of the key formulas presented in \cite{GGS}.
In this paper, we will give the correct formula and develop somewhat
different arguments that lead
to a proof of Theorem~\ref{thmStdN} and help to deduce
Theorem~\ref{thmMoyalN} as well.


\parage
We first give some background and motivation.

Ever since quantum mechanics started in the 1920s,
quantization and semiclassical limit have become
a central theme among a variety of areas in mathematics such as
functional analysis, geometry and topology, representation theory,
pseudo-differential operators and microlocal analysis and symplectic
geometry, to name a few.

Conventional quantum mechanics are formulated in terms of linear
operators on Hilbert space that realize the fundamental Canonical
Commutation Relations (CCR), or of Feynman's path integrals as
conceived by Dirac and developed by Feynman to make the quantum picture
more compatible with the classical one.
Built upon Wigner, Weyl and Groenewold's insights and pioneered by
Moyal, deformation quantization is a third 
formulation, in full phase space, which evolved gradually into an autonomous
theory, with its own complete internal logic, that is conceptually
very appealing.

The idea of deformation quantization is to achieve the Heisenberg CCR
by deforming the commutative algebra of functions on the phase space
(classical observables) to a non-commutative associative algebra.

In \cite{Moyal49}, Moyal defined the so-called Moyal star
product~$\star_M$ 
in relation with statistical properties of quantum mechanics.
For one degree of freedom,
the standard Poisson structure for functions $f(q,p)$ and $g(q,p)$ in~$\R^2$ being
\[
    \{f,g\}=\frac{\partial f}{\partial p}\frac{\partial g}{\partial
    q}-\frac{\partial f}{\partial q}\frac{\partial g}{\partial
    p}
  = \mu\circ P(f\otimes g),
  \quad\text{with}\ens
  P \defeq \frac{\partial }{\partial p}\otimes\frac{\partial
  }{\partial q}-\frac{\partial }{\partial q}\otimes\frac{\partial
  }{\partial p},
\]
where $\mu$ is the usual pointwise product of functions,
the Moyal star product of two classical observables is the formal series in~$t$ obtained as 
\beglab{eqdefstarMP}
  f\star_M g = \mu \circ \exp\Big(\frac{tP}{2}\Big)(f\otimes g)
  = \mu\circ \Big( \ID + \sum_{k\geq 1}\frac{t^k}{2^k k!}P^k \Big) (f\otimes g) 
\edla
with
\[
  P^k 
  %
  = \sum_{n=0}^k (-1)^{k-n} \bino{k}{n} (\pa_p\, \ot\,\pa_q)^n (\pa_q\,\ot\,\pa_p)^{k-n}
  = \sum_{n=0}^k (-1)^{k-n} \bino{k}{n} \pa_p^n\, \pa_q^{k-n} \, \ot \, \pa_p^{k-n}\, \pa_q^{n}.
%
%
\]
Here, $t$ is the deformation parameter, taken to be ${i\hbar}$ 
in quantum mechanics.
When extended to $C^\infty(\R^{2N})[[t]]$, the Moyal star product is a non-commutative associative product.
%
We have $f\star_M g = fg + \frac{t}{2}\{f,g\} + O(t^2)$, hence
\beglab{eqcommutstarMdeform}
  [f,g]_M \defeq \frac{1}{t}(f\star_M g-g\star_M f) = \{f,g\} + O(t),
\edla
we thus recover the
Poisson algebra structure of the classical observables in the limit
$\hbar = t/i \to 0$. 
Moreover 
\[
p\star_M q-q\star_M p = t, \quad\text{\ie}\ens [p,q]_M = 1.
\]
The extension to $N\ge1$ degrees of fredom is obtained by
replacing~$P$ by
\beglab{eqdefPN}
P = \sum_{j=1}^N \Big( \frac{\partial }{\partial p_j}\otimes\frac{\partial
  }{\partial q_j}-\frac{\partial }{\partial q_j}\otimes\frac{\partial
  }{\partial p_j} \Big)
  \edla
  in~\eqref{eqdefstarMP},
  so that~\eqref{eqcommutstarMdeform} still holds. Moreover,
  the CCR are realized: $[p_i,q_i]_M = 1$ and $[p_i,q_j]_M =  [p_i,p_j]_M = [q_i,q_j]_M = 0$ for $i\neq j$.
  
The Moyal star product can be viewed as a non-commutative
associative deformation of the usual product of functions. The idea to
view Quantum Mechanics as a deformation of Classical Mechanics was
promoted by Bayen-Flato-Fronsdal-Lichnerowicz-Sternheimer
\cite{BFFLS} in the 1970s and led to what is now called Deformation
Quantization Theory.

For general symplectic manifolds, the existence of a star product was
proved in \cite{DWL83} and \cite{Fed85}, \cite{F94}.
In particular, Fedosov recursively constructed 
a star product through a canonical flat
connection on the Weyl bundle.

For an arbitrary Poisson structure~$\pi$ on~$\R^N$, Kontsevich
(\cite{K03}) constructed an intriguing explicit formula for its
deformation quantization:
\begin{equation}\label{KontsevichDQ}
  f\star_K g 
  = f g + \sum_{k\ge1}
  t^k\sum_{\Gamma\in G_k} c_\Gamma B_{\Gamma,\pi}(f,g),
\end{equation}
where each~$G_k$ is a suitable collection of graphs, the~$c_\Ga$'s are
universal coefficients, and the $B_{\Ga,\pi}$'s are polydifferential
operators depending on the graph~$\Ga$ and the Poisson
structure~$\pi$.
Recently, a deep connection between these universal
coefficients~$c_\Ga$ and multiple zeta values\footnote{Interestingly,
  multiple zeta values are themselves deeply related to Resurgence
  Theory---see \eg \cite{Wald}.} has been brought to light \cite{BPP}.

For a general Poisson manifold $(M, \pi)$, the existence of a star
product that satisfies the analogue of~\eqref{eqcommutstarMdeform} is
a consequence of the formality theorem which establishes an $L_\infty$
quasi-isomorphism between two differential graded Lie algebras
(DGLAs): the Hochschild complex of the associative algebra
$A=C^\infty(M)$ and its cohomology.


\parage
It is known after Dyson \cite{Dy} and others that, in quantum field
theory, almost all the series in~$\hbar$ describing physical quantities
are divergent and must be interpreted as giving asymptotic information.
In quantum mechanics, this can even be traced back to as early as Birkhoff.
When Voros developed the exact WKB theory \cite{V83} to study the
spectrum of Sturm-Liouville operators, he already conjectured the
resurgent character of these series.
Resurgence Theory was then a new perspective, initiated by \'{E}calle
(\cite{Eca81}, \cite{Eca85}), to deal with asymptotic series.
\'{E}calle immediately clarified and confirmed Voros'
conjecture in \cite{E84} and \cite{Eca85}.
Pham and his collaborators devoted a lot of energy to make the whole
picture complete in \cite{Pham} and a series of papers in the 1990s, culminating in the proof of the
conjectural formula proposed by Zinn-Justin \cite{ZJ} on
multi-instanton expansions in quantum mechanics
(however they had to rely on a resurgence conjecture stated in
\cite{E84}, the proof of which has not yet been given).


In a nutshell, Resurgence Theory deals with formal series $\ti\ph(t) =
\sum_{n\ge0} a_n t^n$ (in applications to physics, the coefficients~$a_n$ may be functions
on the configuration space or the phase space) and
their formal Borel transforms\footnote{%
  \label{footcBvsbe}
  In this article, we depart from the usual convention of Resurgence
  Theory, which would be to define the formal Borel
  transform as
$\cB\big(\sum_{n\ge1} a_n t^n\big)
= \sum_{n\ge1} a_n \frac{\xi^{n-1}}{(n-1)!}$ and to handle the
constant term~$a_0$ separately by setting $\cB(1)=\de$ (a symbol that
can be identified with the Dirac mass at~$0$). Obviously,
formula~\eqref{eqdefBtrsf} yields
$\be(\ti\ph)=\cB(t\ti\ph)$. The advantage of~$\cB$ is that it gives
rise to slightly simpler formula for convolution and Laplace transform. However, in this article, we prefer not to
have to deal separately the $t$-independent term.  
  }%
~$\wh\ph(\xi)$, with
\beglab{eqdefBtrsf}
\text{formal Borel transform} \; \be \col
\ti\ph(t) = \sum_{n\ge0} a_n t^n
\mapsto \wh\ph(\xi) = \sum_{n\ge0} a_n \frac{\xi^n}{n!}
\edla
and imposes convergence of~$\wh\ph(\xi)$ and suitable conditions on
its analytic continuation,
so as to be able to analyse the various ``Borel-Laplace sums''
$\gS^\tht\ti\ph(t) = \int_0^{\eith\infty} \ee^{-\xi/t} \wh\ph(\xi) \dd
t/t$,
for all non-singular directions~$\tht$,
which are all asymptotic to~$\ti\ph(t)$ as $\abs t\to0$ but differ by exponentially
small quantities.
One requires the analytic continuation of the convergent germ~$\wh\ph(\xi)$, roughly speaking, to have at
worse isolated singularities---more precisely, to be
``$\Om$-continuable'' with a certain prescribed set~$\Om$ of potential
singular points \cite{Eca81}, \cite{Sau13};
or ``endlessly continuable'' \cite{CNP93};
or ``continuable without a cut'' \cite{Eca85}
(in order of increasing generality).

Another variant of this property of continuability in the Borel plane
was introduced in \cite{GGS} under the name ``algebro-resurgence'',
that was designed for situations where the coefficients~$a_n$ depend analytically
on affine variables $(q,p)\in\C^{2N}$: the singular locus of
$\wh\ph(\xi,q,p)$ is required to be a proper algebraic subvariety
of~$\C^{1+2N}$, the germ $\wh\ph(\xi,q,p)$ should have analytic
continuation along all the paths that avoid it;
in particular, for fixed~$(q,p)$, only finitely many singular points
can exist in the Borel plane.
It is with this version of resurgence that we will work throughout this article.


\parage
Our initial motivation was to understand Deformation Quantization and
the explicit construction~(\ref{KontsevichDQ}) of Kontsevich from the
viewpoint of Resurgence Theory.
But already at the level of the Moyal star product~\eqref{eqdefstarMP},
even with analytic classical observables that do not depend on~$t$,
one can see that divergence of the series generically occurs, but with at most
factorial growth due to the Cauchy inequalities.
It is thus natural to consider the Moyal star product of two elements
of~$\ti f$ and~$\ti g$ of
$\C\{q_1,\ldots,q_N,p_1,\ldots,p_N\}[[t]]$ 
and to enquire on $\be(\ti f\star_M\ti g)$ in terms of $\be\ti f$
and~$\be\ti g$, \ie to investigate the Borel counterpart of the Moyal star product:
\beglab{eqdefastM}
\wh f *_M\wh g \defeq \be\big( \be\ii\wh f \star_M \be\ii\wh g \big).
\edla

This is what Garay, de Goursac and van Straten did in \cite{GGS}
with the ``standard star product''~$\star_S$ defined by~\eqref{eqdefstarS}, which is a star product
equivalent to the Moyal one (see Section~\ref{secPrelimin}).
They considered
\beglab{eqdefastS}
\wh f *_S\wh g \defeq \be\big( \be\ii\wh f \star_S
\be\ii\wh g \big)
\edla
with a view to proving Theorem~\ref{thmStdN}:
supposing that $\ti f$ and~$\ti g$ are algebro-resurgent series, \ie
that $\wh f$ and~$\wh g$ are algebro-resurgent germs, is it true that
$\wh f*_S\wh g$ is an algebro-resurgent germ (and hence that $\ti
f\star_S\ti g$ is an algebro-resurgent series)?

However, the analysis in \cite{GGS} relies on an integral
representation of~$*_S$ that is flawed, thus invalidating the key
Proposition~3.3 and the purported proof of
Theorem~\ref{thmStdN} in that article.
In Section~\ref{secPrelimin}, we will give another integral
representation of~$*_S$, formula~\eqref{formulahigh}.
The correct formula is more intricate than that of \cite{GGS}; therefore, following the
analytic continuation of $\wh f*_S\wh g$ (where both factors are
supposed to be algebro-resurgent) requires considerably more work.

For the sake of clarity, we will begin in Section~\ref{secPrelimin}
with the case of one degree of freedom and give in Lemma~\ref{integral
  rep of Borel star} the formula for~$*_S$ for that case.
It is a mixture of convolution\footnote{%
  \label{footconvol}
  ``Convolution'' is the operation that corresponds to the
  multiplication of formal series via formal Borel transform. Beware
  that Resurgence Theory usually makes use of the formula
  corresponding to~$\cB$, rather than~$\be$, in accordance with Footnote~\ref{footcBvsbe}.
  The image by~$\cB$ of the product $(\cB\ii\wh\ph) (\cB\ii\wh\psi)$ is
  the function $\int_0^\xi \wh\ph(\xi_1)\wh\psi(\xi-\xi_1)\,\dd\xi_1$,
  whereas the image by~$\be$ of $(\be\ii\wh\ph) (\be\ii\wh\psi)$ is
  $\frac{\pa}{\pa\xi}\big(\int_0^\xi \wh\ph(\xi_1)\wh\psi(\xi-\xi_1)\,\dd\xi_1\big)$.
  }
  and Hadamard product (which has the classical integral
  representation~\eqref{eqmultconv}); more specifically, the formula involves the Hadamard
  product \wrt~$\ze$ of the Taylor expansions
$\wh f(\xi_1,q,p+\ze)\odot \wh g(\xi_2,q+\ze,p)$
and then a convolution-like integration \wrt~$\xi_1$ and~$\xi_2$.

Analytic continuation of convolution is a classical topic in
Resurgence Theory \cite{Eca81}, \cite{CNP93}, \cite{Sau13},
\cite{S14}.
We will adapt these techniques to our more intricate situation in
Section~\ref{proofconvo}.
The analytic continuation of the Hadamard product of two
$\Om$-continuable germs has been treated in \cite{LSS}, with a
possibly infinite singular locus~$\Om$;
our situation is simpler inasmuch as it involves only finite singular
loci in the Borel plane, as we will see in Section~\ref{proofHa}
devoted to the Hadamard part of the formula for~$*_S$.

The technique for following the analytic continuation of $\wh f*_S\wh
g$ in the case of $N$ degrees of freedom is sketched
 in Section~\ref{secHigher}.
This will lead us to a proof of Theorem~\ref{thmStdN} that follows a
path rather different than that of \cite{GGS}.
Using the concrete form of the equivalence between the Moyal and
standard star products $\star_M$ and~$\star_S$, we will be able to
relate~$*_M$ and~$*_S$ by an integral transform, and deduce that
\begin{quote}
  $\wh f$ and~$\wh g$ algebro-resurgent germs $\Imp$
  $\wh f*_M\wh g$ algebro-resurgent germ,
\end{quote}
which is equivalent to Theorem~\ref{thmMoyalN}.

Hence, algebro-resurgent series form a subalgebra of
$\big(\C\{q_1,\ldots,q_N,p_1,\ldots,p_N\}[[t]],\star_S \big)$
or of
$\big(\C\{q_1,\ldots,q_N,p_1,\ldots,p_N\}[[t]],\star_M \big)$.


\parage
The paper is organized as follows.
\begin{enumerate}[--]
\item Section~\ref{secPrelimin} deals with definitions, examples and
elementary properties for the Moyal and standard star products,
$\star_M$ and~$\star_S$, and their Borel counterparts~$*_M$
and~$*_S$.
It also contains the integral representation formulas that will be
used in the rest of the article.
\item Section~\ref{secAlgRes} deals with the definition of
  algebro-resurgent series and algebro-resurgent germs, and states
  three lemmas that are instrumental in our proof of
  Theorems~\ref{thmMoyalN} and~\ref{thmStdN}.
\item Section~\ref{secSimpPoly} introduces the notion of a
  multivariate polynomial that is ``simple \wrt\ one of its
  variables'', as an algebraic preparation to handle more
  conveniently the smooth algebraic varieties which appear in the singular
  loci in the Borel plane.
\item Section~\ref{proofconvo} deals with the ``convolution part'' of
  our formula for~$*_S$.
\item Section~\ref{proofHa} deals with the ``Hadamard part'' of the formula.
\item Section~\ref{secHigher} explains how to adapt the proof from
  $N=1$ to $N$ arbitrary.
  %
%
\end{enumerate}

\newpage

\section{Borel counterparts of the Moyal and standard star products}   \label{secPrelimin}

\parage
The star product~$\star_M$ is defined by formula~\eqref{eqdefstarM} if
$N=1$, or~\eqref{eqdefstarMP} with $P$ as in~\eqref{eqdefPN} for
general $N\ge1$.
These formulas make sense
in $\C\{q_1,\ldots,q_N,p_1,\ldots,p_N\}[[t]]$ as well as in
\[
\ti\cQ_{2N+1} \defeq \C[[q_1,\ldots,q_N,p_1,\ldots,p_N]][[t]] = \C[[t,q_1,\ldots,q_N,p_1,\ldots,p_N]].
\]
The same is true for~$\star_S$, which is defined by the
formula~\eqref{eqdefstarS} or, equivalently
\beglab{eqdefstarSP}
  f\star_S g = \mu \circ \exp\bigg( t \sum_{j=1}^N \Big( \frac{\partial }{\partial p_j}\otimes\frac{\partial
  }{\partial q_j} \Big) \bigg)(f\otimes g).
\edla
Recall that the formal deformation parameter is $t = {i\hbar}$.



It is well-known that~$\star_S$ and~$\star_M$ are equivalent, in the
sense that there is a ``transition operator''~$T$ mapping the former
to the latter: 
$T(f\star_S g) = (Tf)\star_M(Tg)$.
It is sufficient to take
\beglab{eqdefequivT}
  T \defeq \exp\Big(-\frac{i\hbar}{2}\sum_{1\le j \le
    N}\partial_{q_j}\partial_{p_j}
  \Big) =
  \exp\Big(-\frac{t}{2}\sum_{1\le j \le
    N}\partial_{q_j}\partial_{p_j}
  \Big),
\edla
the inverse of which is given by $T\ii = \exp\big( \frac{t}{2}\sum
\partial_{q_j}\partial_{p_j} \big)$.
In other words, we have
\[
  T \col \big(\ti\cQ_{2N+1},\star_S\big) \to \big(\ti\cQ_{2N+1},\star_M\big)
  \ens\;\text{isomorphism of associative algebras.}
\]
It is with~$\star_S$ that we will work most of the time, because the
formulas are simpler with it than with~$\star_M$, hence we use the
abbreviation
\begin{notation}
  From now on, we set $\star = \star_S$ for the standard star product,
  and $* = *_S$ for its Borel counterpart~\eqref{eqdefastS}.
  We will call~$*$ the ``Borel-star product''.
\end{notation}

\begin{example} Here is a simple example:
  \beglab{eqsimplex}
  (tp)\star(tq) = t^2pq + t^3, \ens
  (tq)\star(tp) = t^2pq, \ens
  (tp)\star_M(tq) = t^2pq + t^3, \ens
  (tq)\star_M(tp) = t^2pq - t^3.
  \edla
  Note that $T(t^2pq) = t^2pq - \frac{t^3}{2}$.
  \end{example}




\parage
We are mostly interested in the subspace
$\C\{q_1,\ldots,q_N,p_1,\ldots,p_N\}[[t]]$ of $\ti\cQ_{2N+1}$.
However, it is important to realize that we cannot restrict ourselves
to the too narrow subspace\footnote{However, the even narrower
  subspace of polynomials $\C[t,q_1,\ldots,q_N,p_1,\ldots,p_N]$ is
  stable under~$\star_M$ and~$\star$.}
$\C\{t,q_1,\ldots,q_N,p_1,\ldots,p_N\}$ consisting of formal series
which converge in a neighbourhood of the origin in~$\C^{2N+1}$,
because even if~$f$ and~$g$ do not depend on~$t$, their star product
may be divergent.
Here is a simple example taken from \cite{GGS}, and a variant:
\begin{examples}
  The geometric series $(1-p)\ii$ and $(1-q)\ii$ give rise to a
  divergent series
  \beglab{eqEulerseries}
  (1-p)\ii \star (1-q)\ii = \sum_{k\ge 0} k! t^k \big( (1-p)(1-q)
  \big)^{-k-1}.
  \edla
  The logarithm series $\log(1-p)$ and $\log(1-q)$ give rise to a
  divergent series
  \beglab{eqlogstarlog}
  \log(1-p) \star \log(1-q) = \log(1-p) \log(1-q) + \sum_{k\ge1}
  \frac{(k-1)!}{k} t^k \big( (1-p)(1-q)
  \big)^{-k}.
  \edla
\end{examples}

Note however the $1$-Gevrey character \wrt~$t$ of these examples: the
coefficient of~$t^k$ essentially has at most factorial growth, hence
convergence is restored when~$t^k$ is replaced by $\xi^k/k!$, \ie
their image by the formal Borel transform~\eqref{eqdefBtrsf} belongs
to the space of convergent series $\C\{\xi,q,p\}$. This is a general phenomenon:
extending the definition of the formal Borel transform by the formula
\beglab{eqdefBtrsfGEN}
\be \col
\ti\ph = \sum_{n\ge0} a_n(z_1,\ldots,z_r) t^n \in \C[[t,z_1,\ldots,z_r]]
\mapsto \wh\ph = \sum_{n\ge0} a_n(z_1,\ldots,z_r) \frac{\xi^n}{n!}
\in \C[[\xi,z_1,\ldots,z_r]],
\edla
we call $1$-Gevrey formal series \wrt~$t$ the elements of
\begla
\ti\cQ_{r+1} \defeq \be\ii\big( \wh\cQ_{r+1} \big) \subset \C[[t,z_1,\ldots,z_r]],
\quad
\text{with} \ens
\wh\cQ_{r+1} \defeq \C\{\xi,z_1,\ldots,z_r\}.
\edla
and we have, as noted in \cite{GGS},
\begin{theorem}
  The subspace $\ti\cQ_{2N+1}$ is stable under the Moyal star product~$\star_M$, as
  well as under the standard product $\star=\star_S$.
\end{theorem}
The proof is a consequence of~\eqref{implicwhcQr} in Lemma~\ref{lemBorelOperators}.
\medskip

In the case of Examples~\eqref{eqEulerseries} and~\eqref{eqlogstarlog}, we find
\begin{gather}
  \be\big( (1-p)\ii \star (1-q)\ii \big) = \Big( 1 - \xi \big(
  (1-p)(1-q) \big)\ii \Big)\ii \\[1ex]
  \be\big( \log(1-p) \star \log(1-q) \big) = \log(1-p) \log(1-q) + 
  \Li_2\Big(\xi \big( (1-p)(1-q) \big)\ii \Big),
\end{gather}
where
\beglab{eqdefLid}
\Li_2(z) = \sum_{k\ge1} \frac{z^k}{k^2}
= - \int_0^z \frac{\log(1-\ze)}{\ze}\,\dd\ze.
\edla
In fact, the divergent series~\eqref{eqEulerseries} is essentially the
famous Euler series, the most elementary example of resurgent series.

\parage
Recall that we have defined the Borel counterpart~$*_M$ of
the Moyal star product~$\star_M$ by~\eqref{eqdefastM}, and the
Borel-star product $*=\star_S$, 
counterpart of the standard star product~$\star$,
by~\eqref{eqdefastS}.

For the sake of simplicity, we begin with the case of one degree of
freedom.


\begin{lemma}\label{integral rep of Borel star}
  There is an integral representation of the Borel-star product in
  $\C[[\xi,q,p]]$ as follows:
\begin{equation}\label{rightone}
  \hat{f}\ast\hat{g}(\xi,q,p)=\frac{d^3}{d\xi^3}\int_0^\xi
  d\xi_1\int_0^{\xi-\xi_1}d\xi_2\int_0^{\xi-\xi_1-\xi_2}d\xi_3 
  \int_0^{2\pi}\frac{d\theta}{2\pi}
  \hat{f}(\xi_1,q,p+\sqrt{\xi_3}e^{-i\theta})\hat{g}(\xi_2,q+\sqrt{\xi_3}e^{i\theta},p),  
\end{equation}
where the integrand is considered as element of $\mathbb{C}[e^{\pm i\theta}][[q,p,\xi_1,\xi_2,\sqrt{\xi_3}]]$ and integration in $\theta$ is performed termwise.

Moreover, if both factors are convergent, \ie $\hat f,\hat g \in
\wh\cQ_3$, then so is their Borel-star product: $\hat f * \hat g \in\wh\cQ_3$.
\end{lemma}

\begin{proof} We expand $\hat{f}$ and $\hat{g}$:
\[\begin{split}
\hat{f}\ast\hat{g}(\xi,q,p)
=&
\sum\limits_{m,n,s \geq 0} \frac{n!m!}{(n+m+s)!s!}(\partial_p^sf_m) (\partial_q^sg_n)\xi^{m+n+s}
           \\
=&\frac{d^3}{d\xi^3}\int_0^\xi d\xi_1\int_0^{\xi-\xi_1}d\xi_2\int_0^{\xi-\xi_1-\xi_2}d\xi_3
 \sum\limits_{n,m,s \geq 0} \frac{\partial_p^sf_m \partial_q^sg_n}{s!s!}\xi_1^m \xi_2^n \xi_3^s.
\end{split}\]
Using the fact that the Hadamard product of two formal series $\phi(\xi) = \sum\limits_{s\geq 0} a_s \xi^s$ and $\psi(\xi) = \sum\limits_{t\geq 0} b_t \xi^t $ can be written as $\phi \odot \psi (\xi) =\frac{1}{2\pi } \int _0 ^{2\pi } \phi(\sqrt{\xi} e^{-i\theta}) \psi(\sqrt{\xi} e^{i\theta}) d\theta$ with termwise integration in $\mathbb{C}[e^{\pm i\theta}][[\sqrt{\xi}]]$, we get
\[\begin{split}
\hat{f}\ast\hat{g}(\xi,q,p)=&\frac{d^3}{d\xi^3}\int_0^\xi d\xi_1\int_0^{\xi-\xi_1}d\xi_2\int_0^{\xi-\xi_1-\xi_2}d\xi_3
           \\
&\frac{1}{2\pi}\int_0^{2\pi}d\theta\sum\limits_{m,n \geq 0} f_m(q,p+\sqrt{\xi_3}e^{-i\theta}) g_n(q+\sqrt{\xi_3}e^{i\theta},p) \xi_1^m \xi_2^n
           \\
=&
\frac{d^3}{d\xi^3}\int_0^\xi d\xi_1\int_0^{\xi-\xi_1}d\xi_2\int_0^{\xi-\xi_1-\xi_2}d\xi_3
           \\
&\frac{1}{2\pi}\int_0^{2\pi}d\theta\hat{f}(\xi_1,q,p+\sqrt{\xi_3}e^{-i\theta})\hat{g}(\xi_2,q+\sqrt{\xi_3}e^{i\theta},p) .
\end{split}\]
%


%
%
%





Now, suppose $\hat{f}, \hat{g} \in \wh\cQ_3 = \mathbb{C}\{\xi,q,p\}$.
The right hand side of (\ref{rightone}) involves a function
$G(\xi_1,\xi_2,s,q,p):=\int_0^{2\pi} 
\hat{f}(\xi_1,q,p+se^{-i\theta})
\hat{g}(\xi_2,q+se^{i\theta},p) \frac{d\theta}{2\pi}$ which clearly belongs to $\mathbb{C}\{\xi_1,\xi_2,s ,q,p\}$ (indeed, we can take $(s,\xi_1,\xi_2, q,p)$ close enough to origin). In fact, $G(\xi_1,\xi_2,s,q,p)=G(\xi_1,\xi_2,-s,q,p)$, hence $F(\xi_1,\xi_2,\xi_3,q,p):=G(\xi_1,\xi_2,\sqrt{\xi_3},q,p)\in \mathbb{C}\{\xi_1,\xi_2,\xi_3,q,p\}$ and the right hand side of (\ref{rightone}) can be written as
\begin{equation}
  \frac{d^3}{d\xi^3}\int_0^\xi d\xi_1\int_0^{\xi-\xi_1}d\xi_2\int_0^{\xi-\xi_1-\xi_2}d\xi_3 F(\xi_1,\xi_2,\xi_3,q,p),
\end{equation}
hence it defines a holomorphic germ in $\mathbb{C}\{\xi,q,p\}$.
\end{proof}

The integral formula~\eqref{rightone} differs from the one given in
Proposition~3.3 of \cite{GGS}, which is not correct.
Take for instance $\hat f = \xi p$ and $\hat g = \xi q$: we know
by the first equation in~\eqref{eqsimplex} that we must find
\begla
(\xi p) * (\xi q) = pq \frac{\xi^2}{2!} + \frac{\xi^3}{3!},
\edla
and the reader may check that our formula produces the right outcome, but not the
formula from \cite{GGS}, which yields a term $\frac{\xi^3}{2!}$
instead of $\frac{\xi^3}{3!}$.

\begin{remark}
  Instead of writing the Hadamard product $\phi\odot\psi(\xi)$ as we
  did in our proof, we could have used the integration variable
  $\ze=\sqrt\xi \eith$ and then the Cauchy theorem, which yields
\begla
\phi\odot \psi(\xi)=\frac{1}{2\pi i}\oint_C
\phi(\frac{\xi}{z})\psi(z)\frac{dz}{z}
\quad\text{with any circle}\ens
C\col \tht \mapsto c\, \eith
\ens\text{of radius}\ens
c \in \big( \tfrac{|\xi|}{R_\psi}, R_\phi \big),
\edla
where $R_\phi$ and~$R_\psi$ are the radii of convergence of~$\phi$
and~$\psi$.
Correspondingly, Formula~\eqref{rightone} can be rewritten
\begla
  \hat{f}\ast\hat{g}(\xi,q,p)=\frac{d^3}{d\xi^3}\int_0^\xi
  d\xi_1\int_0^{\xi-\xi_1}d\xi_2\int_0^{\xi-\xi_1-\xi_2}d\xi_3 
  \oint_C \frac{d z}{2\pi i z}
  \hat{f}(\xi_1,q,p+\frac{\xi_3}{z}) \hat{g}(\xi_2,q+z,p), 
\edla
where~$C$ is an appropriate circle.
\end{remark}

\begin{lemma}   \label{integral rep of Borel starM}
If $\hat{f},\hat{g} \in \C[[\xi,q,p]]$, then
\begin{equation}
\begin{split}
  \hat{f} *_M \hat{g} 
                         &= \frac{d^4}{d\xi^4} \int^\xi_0 d\xi_1 \int^{\xi-\xi_1} _0 d\xi_2 \int^{\xi-\xi_1-\xi_2}_0 d\xi_3 \int^{\xi-\xi_1-\xi_2-\xi_3}_0 d\xi_4   \\
                         & (\frac{1}{2\pi i})^2 \oint_{C_1} dz_1\oint_{C_2} dz_2 \hat{f}(\xi_1,q+z_1,p+z_2) \hat{g} (\xi_2, q+\frac{\xi_4}{2z_2},p-\frac{\xi_3}{2z_1})
                       \end{split}
                     \end{equation}
                     with integration on appropriate circles~$C_1$ and~$C_2$.

                       Moreover,  if both factors are convergent, \ie $\hat f,\hat g \in
\wh\cQ_3$, then so is $\hat f *_M \hat g$.
\end{lemma}

\begin{proof}
\[\begin{split}
  \text{RHS}=& \frac{d^4}{d\xi^4} \int^\xi_0 d\xi_1 \int^{\xi-\xi_1} _0 d\xi_2 \int^{\xi-\xi_1-\xi_2}_0 d\xi_3 \int^{\xi-\xi_1-\xi_2-\xi_3}_0 d\xi_4 \\
  & \sum\limits_{n,m}   \Big(\partial_p^n \partial_q^m\hat{f} (\xi_1,q,p)\Big) \Big( \partial_p^m\partial_q^n\hat{g}(\xi_2,q,p)\Big) \frac{(-1)^m}{2^{m+n}m!^2 n!^2 } \xi_3^m \xi_4^n  \\
  =&\frac{d^2}{d\xi^2} \int^\xi_0 d\xi_1 \int^{\xi-\xi_1} _0 d\xi_2  \\
  & \sum\limits_{n,m,\alpha,\beta} \Big( \partial_p^n \partial_q^m f_\alpha(q,p) \xi_1 ^{\alpha}\Big)\Big( \partial_p^m\partial_q^n g_\beta(q,p) \xi_2^\beta\Big) \frac{(-1)^m}{2^{m+n}m! n!(n+m)!} (\xi-\xi_1-\xi_2)^{n+m}  \\
  =& \text{LHS}.
\end{split}
  \]
\end{proof}

The last statement in Lemma~\ref{integral rep of Borel starM} can also
be derived from the following integral representations of the Borel
counterparts of~$T$ and~$T\ii$,
\begla
\wh T\hat f \defeq \be T \be\ii \hat f,
\qquad
\wh T\ii\hat f = \be T\ii \be\ii \hat f.
\edla

\begin{lemma}   \label{integral rep of Borel T}
    For any $\hat f \in \C[[\xi,q,p]]$,
  \begin{equation}
  \wh T\hat f = \frac{d}{d\xi} \int_0^{\xi} d\xi_1 \frac{1}{2\pi i}\oint_C \hat{f}(\xi-\xi_1,q+z,p-\frac{\xi_1}{2z}) \frac{dz}{z},
\end{equation}
\begin{equation}
  \wh T\ii\hat{f} = \frac{d}{d\xi} \int_0^{\xi} d\xi_1 \frac{1}{2\pi i}\oint_C \hat{f}(\xi-\xi_1,q+z,p+\frac{\xi_1}{2z}) \frac{dz}{z},
\end{equation}
with integration on appropriate circle~$C$.

Moreover, if $\hat f$ is convergent, \ie $\hat f\in
\wh\cQ_3$, then so is $\wh T\hat f$.
\end{lemma}

\begin{proof}
  \[\begin{split}
  \text{RHS}&= \frac{d}{d\xi} \int_0^{\xi} d\xi_1 \frac{1}{2\pi i} \oint_C \sum\limits_{n,m} \frac{\partial^n_q \partial^m_p \hat{f}(\xi-\xi_1,q,p)}{n!m!}{(-\frac{\xi_1}{2})^m} z^{n-m-1} dz \\
     &= \frac{d}{d\xi} \int_0^{\xi} d\xi_1 \sum\limits_{n} \frac{\partial^n_q \partial^n_p \hat{f}(\xi-\xi_1,q,p)}{n!n!}(-\frac{\xi_1}{2})^n  \\
     &= \beta\big( \sum\limits_{n} \frac{\partial^n_q \partial^n_p \tilde{f}(t,q,p)}{n!} \cdot (-\frac{t}{2})^n  \big)  \\
     &= \text{LHS}.
  \end{split}
  \]
\end{proof}
  

\parage
We now consider the case of an abitrary number of degrees of freedom, say~$r$.
We set $q=(q_1,\cdots,q_r)$ and $p=(p_1,\cdots,p_r)$.
If $\hat{f}(\xi,q,p)=\sum\limits_{m=0}^{\infty} f_m(q,p) \xi^m$, $\hat{g}(\xi,q,p)=\sum\limits_{n=0}^{\infty} g_n(q,p) \xi^n$, then
\[
  \hat{f}\ast \hat{g} (\xi,q,p) = \sum_{m,n,k_1,\cdots,k_r\ge0}
  \frac{n!m!}{k_1!\cdots k_r!} (\partial_{p_1}^{k_1} \cdots
  \partial_{p_r}^{k_r} f_m)( \partial_{q_1}^{k_1} \cdots
  \partial_{q_r}^{k_r} g_n )
  \frac{\xi^{k_1+\cdots+k_r+n+m}}{(k_1+\cdots+k_r+n+m)!}.
  \]

  \begin{lemma}   \label{lemBorelOperators}
    There are integral representation formulas analogous to those of
    Lemmas~\ref{integral rep of Borel star},
    \ref{integral rep of Borel starM}
  and~\ref{integral rep of Borel T} for the case of $r$ degrees of freedom.
  For instance, for $\hat f, \hat g \in
  \C[[\xi_1,q_1,\ldots,q_r,p_1,\ldots,p_r]]$, the formula
  generalising~\eqref{rightone} is
  \begin{equation}\label{formulahigh}
  \begin{split}
  \hat{f}\ast \hat{g} (\xi,q,p) =&
  \frac{d^{r+2}}{d \xi ^{r+2}} \int^\xi_0 d\xi_1
  \int^{\xi-\xi_1}_0 d\xi_2
  \cdots \int_0^{\xi-\xi_1-\cdots -\xi_{r+1}} d\xi_{r+2}
  \big( \frac{1}{2\pi} \big)^r \overbrace{\int_0^{2\pi} \cdots \int_0^{2\pi}}^r d\theta_1 \cdots d\theta_r  \\
  &\hat{f}(\xi_{r+1}, q_1,\cdots, q_r,p_1+\sqrt{\xi_1}e^{-i\theta_1}, \cdots, p_r+\sqrt{\xi_r} e^{-i\theta_r}) \\
  & \hat{g}(\xi_{r+2}, q_1+\sqrt{\xi_1}e^{i\theta_1}, \cdots, q_r+\sqrt{\xi_r} e^{i\theta_r}, p_1, \cdots ,p_r).
\end{split}
\end{equation}

These formulas entail that
\beglab{implicwhcQr}
\hat f, \hat g \in \wh\cQ_{2r+1} \Imp
\hat f*\hat g,\;\; \hat f*_M\hat g,\;\; \wh T\hat f,\;\; \wh T\ii\hat f \in \wh\cQ_{2r+1} .
\edla
\end{lemma}


\section{Algebro-resurgent germs}   \label{secAlgRes}

\begin{notation}
  We shall use the following notation:
  \begin{equation}
    \mathbb{D}_{\tau} :=\{ z \in \mathbb{C} \mid |z| < \tau\}.
  \end{equation}
\end{notation}

We know that if $f \in \mathbb{C}\{z_1,\cdots,z_n\}$, then there
exists $\tau>0$, such that $f$ is the germ of a function holomorphic
in the polydisc $\mathbb{D}_\tau^n \subset \mathbb{C}^n$.
Following \cite{GGS}, we set the

\begin{definition}
The set of algebro-resurgent germs is defined as follows: for any non-negative integer $n$,
\begin{equation}
\begin{split}
{\widehat{\mathcal{Q}}}^{\mathcal{A}}_{n+1}:=
\Big\{ &f\in\mathbb{C}\{\xi,z_1,\cdots,z_n\}
\mid \exists\text{ proper algebraic subvariety }V\subset\mathbb{C}^{n+1}, \\
&s.t.\ f\ \text{admits analytic} \text{ continuation along any $C^1$ path $\gamma$ }\\
&\text{contained in  $\mathbb{C}^{n+1}- V$ and having initial
  point $\gamma(0)$ close enough to~$0$}
\Big\}.
\end{split}
\end{equation}
Here ``$\gamma(0)$ close enough to~$0$'' means that
$\gamma(0)\in\mathbb{D}_\tau^{n+1}$ where $\mathbb{D}_\tau^{n+1}$ is a
polydisc where $f$ induces a holomorphic function. \ \\

We denote by $\widehat{\mathcal{Q}}^{\mathcal{A}}$ the disjoint union of $\widehat{\mathcal{Q}}^{\mathcal{A}}_i$ :
\[\widehat{\mathcal{Q}}^{\mathcal{A}}:=\bigsqcup\limits_{i\in\mathbb{N}^\ast}
\widehat{\mathcal{Q}}^{\mathcal{A}}_i
.\]

We define $\widetilde{\mathcal{Q}}^{\mathcal{A}}$ and $\widetilde{\mathcal{Q}}^{\mathcal{A}}_i$ the Borel inverse of $\widehat{\mathcal{Q}}^{\mathcal{A}}$ and $\widehat{\mathcal{Q}}^{\mathcal{A}}_i$ correspondingly:
\[\widetilde{\mathcal{Q}}^{\mathcal{A}}
:=\beta^{-1}(\widehat{\mathcal{Q}}^{\mathcal{A}}), \indent
\widetilde{\mathcal{Q}}^{\mathcal{A}}_i
:=\beta^{-1}(\widehat{\mathcal{Q}}^{\mathcal{A}}_i).\]
\end{definition}


The $1$-degree-of-freedom version of Theorem~\ref{thmStdN} can be
rephrased as

\begin{theorem}   \label{thmrephrased}
\quad If $\ti{f}(t,q,p),\ \ti{g}(t,q,p)\in \widetilde{\mathcal{Q}}^{\mathcal{A}}_3$, then
$\ti{f}\star\ti{g}(t,q,p)\in \widetilde{\mathcal{Q}}^{\mathcal{A}}_3$.
\smallskip

Equivalently,
\begla
\hat{f}(t,q,p),\ \hat{g}(t,q,p)\in \wh{\mathcal{Q}}^{\mathcal{A}}_3
\imp
\hat{f}*\hat{g}(t,q,p)\in \wh{\mathcal{Q}}^{\mathcal{A}}_3.
\edla
\end{theorem}

\medskip

Sections~\ref{secSimpPoly}--\ref{proofHa} are devoted to the proof of
Theorem~\ref{thmrephrased}.
(Then Section~\ref{secHigher} will show how to deduce
Theorem~\ref{thmStdN}, and also  Theorem~\ref{thmMoyalN}.)
Using Formula (\ref{rightone}), %
the proof will be divided into the following three lemmas.\ \\

\begin{lemma}\label{Hadalemma}
\quad If $\hat{f}(\xi,q,p),\hat{g}(\xi,q,p)\in \widehat{\mathcal{Q}}^{\mathcal{A}}_3$, then
\begin{equation}\label{Hadafor}
F(\xi_1,\xi_2,\xi_3,q,p):=\frac{1}{2\pi}\int_0^{2\pi}
\hat{f}(\xi_1,q,p+\sqrt{\xi_3}e^{-i\theta})
\hat{g}(\xi_2,q+\sqrt{\xi_3}e^{i\theta},p)d\theta
\end{equation}
$\in\widehat{\mathcal{Q}}^{\mathcal{A}}_5$.\ \\
\end{lemma}

The proof will be found in section \ref{proofHa}, which is treated as a Hadamard product part in the formula (\ref{rightone}). \ \\

\begin{lemma}\label{convolemma}
\quad If $F(z_1,\cdots,z_n)\in\widehat{\mathcal{Q}}^{\mathcal{A}}_n$, then
\begin{equation}\label{convofor}
f(z,z_2,\cdots,z_n):=\int_0^{\overline{P}(z,z_2,\cdots,z_n)}F(z_1,z_2,\cdots,z_n)dz_1
\end{equation}
$\in\widehat{\mathcal{Q}}^{\mathcal{A}}_n$, where $\overline{P}$ is a polynomial of $n$ variables and $\overline{P}(0,\cdots,0)=0$.\ \\
\end{lemma}

The prove will be found in section \ref{proofconvo}, which is treated as a ``convolution product'' part in the formula (\ref{rightone}).

\begin{lemma}\label{convolemma2}
\quad If $F(z_1,\cdots,z_n)\in\widehat{\mathcal{Q}}^{\mathcal{A}}_n$, then
\begin{equation}
f(z_2,\cdots,z_n):=\int_0^{\overline{P}(z_2,\cdots,z_n)}F(z_1,z_2,\cdots,z_n)dz_1
\end{equation}
$\in\widehat{\mathcal{Q}}^{\mathcal{A}}_{n-1}$, where $\overline{P}$ is a polynomial of $n-1$ variables and $\overline{P}(0,\cdots,0)=0$.\ \\
\end{lemma}

Lemma \ref{convolemma2} follows almost directly from lemma \ref{convolemma}, as will shown at the end of section \ref{proofconvo}. \ \\
\ \\


\section{Simple polynomials with respect to a variable}   \label{secSimpPoly}

In this section, we shall work in $\mathbb{C}^n$ with variables $z_1,\cdots,z_n$ and give the definition of $z_1$-simple polynomial. %
The proposition \ref{propositionz1} is very useful in the following sections and we will prove it carefully. The reason we use the definition `$z_1$-simple polynomial' is that we want the set (\ref{formulaz1simple}) is non-trivial. %
We start with
\[
P(z_1,\cdots,z_n) = \sum\limits _{i=0} ^M
b_i(z_2,\cdots,z_n) z_1^i
\in \C[z_1,\cdots,z_n]=\C[z_2,\cdots,z_n][z_1],
\]
where $b_i(z_2,\cdots,z_n)$'s are polynomials of variables $z_2,\cdots,z_n$ and $b_M \neq 0$.
We denote by $\mathbb{F}$ the fraction field of $\C[z_2,\cdots,z_n]$ and $\ov {\mathbb{F}}$ the algebraic closure of $\mathbb{F}$. Thus,  $P(z_1,\cdots,z_n)$ can be written as
\begin{equation}\label{formulaFz1}
  b_M(z_2,\cdots,z_n) \prod\limits_{\alpha=1}^M \big(z_1-\om_{\alpha}(z_2,\cdots,z_n) \big)
\end{equation}
with $\om_{\alpha}(z_2,\cdots,z_n) \in \overline{\mathbb{F}}$. \ \\
\begin{definition}
\quad Given a non-zero polynomial $F(z_1,\cdots,z_n)\in \C[z_1,\cdots,z_n]$ and the representation of it (formula (\ref{formulaFz1})), $F$ is called $z_1$-simple polynomial if for any $\alpha_1,\alpha_2$, $1\leq \alpha_1<\alpha_2 \leq M$, we have $\om_{\alpha_1}(z_2,\cdots,z_n) \neq \om_{\alpha_2}(z_2,\cdots,z_n) $. Specially, $F$ is $z_1$-simple polynomial if the order of $z_1$ in $F$ is zero.
\end{definition}

\begin{proposition}\label{propositionz1}
  Any proper algebraic subvariety $V$ of $ \C^n$ can be written as $V= \bigcap\limits_{J=1}^K P_J^{-1}(0)$, where $K$ is a positive integer and $P_1,\cdots,P_k$ are $z_1$-simple polynomials.
\end{proposition}

\begin{proof}
Hilbert's basis theorem states that every algebraic variety can be described as a common zero locus of finitely many polynomials. Thus we assume
\[V=\bigcap\limits_{J=1}^K V^J,
\quad
V^J:=\{(z_1,\cdots,z_n)\in\mathbb{C}^n\mid Q^J(z_1,\cdots,z_n)=0\},\text{ for } J=1,\cdots,K , \]
where $Q^J,J=1,\cdots,K$, are non-zero polynomials over $\mathbb{C}^n$. What we want to prove is, for each $Q^J$%
, there exists a non-zero $z_1$-simple polynomial $P^J$ s.t.
 \[Q^{-1}(0)= P^{-1}(0).\]
We will use the abridge notations $Q$ or $P$ later.
Suppose $Q=\sum\limits _{i=0} ^{M}
b_i(z_2,\cdots,z_n) z_1^i$ with $b_i(z_2,\cdots,z_n)$'s are polynomials of variables $z_2,\cdots,z_n$ and $b_M$ non-zero polynomial,
then it has the following factorization in $\ov {\mathbb{F}}[z_1]$:
\[b_M(z_2,\cdots,z_n) \prod\limits_{\alpha=1}^{N} \big(z_1-\om_{\alpha}(z_2,\cdots,z_n) \big)^{s_{\alpha}},\]
where $\om_{\alpha}\in\ov{\mathbb{F}}$,  $\om_{\alpha_1}(z_2,\cdots,z_n) \neq \om_{\alpha_2}(z_2,\cdots,z_n) $ for $1\leq \alpha_1<\alpha_2 \leq N$, integer multiplicities $s_{\alpha}\geq 1$ and $\sum_{\alpha=1}^{N} s_{\alpha}=M$.
Let us suppose that for some $\alpha$, $s_{\alpha}>1$ (if not, the proof is trivial).
We shall use the following notation:
\[\begin{split}
R(z_1,\cdots,z_n)
&\defeq \frac{Q(z_1,\cdots,z_n)}{b_M(z_2,\cdots,z_n)} = \prod\limits_{\alpha=1}^{N} \big(z_1-\om_{\alpha}(z_2,\cdots,z_n) \big)^{s_{\alpha}}, \\
\widetilde R(z_1,\cdots,z_n)
&\defeq \prod\limits_{\alpha=1}^{N} \big(z_1-\om_{\alpha}(z_2,\cdots,z_n) \big).
\end{split}\]
First, we shall prove $\widetilde R(z_1,\cdots,z_n) \in \mathbb{F}[z_1]$. In fact, $R(z_1,\cdots,z_n)$ is reducible in $\mathbb{F}[z_1]$ (irreducible polynomials are separable polynomials). If we consider the minimal polynomial of each root $\om_i(y)$, with Abel's irreducibility theorem, then we get:
\[
\widetilde R = R_1 \cdots R_m, \qquad
R = (R_1)^{\sig_1} \cdots {(R_m)}^{\sig_m}
\]
with $R_1,\ldots,R_m \in \mathbb{F}[z_1]$ and $\sig_i$'s are chosen from
$\{s_1,\ldots,s_{N}\}$.
The idea would be to construct inductively $R_1$ as the minimal
polynomial in $\mathbb{F}[z_1]$ of $\om_1\in\ov {\mathbb{F}}$, then $\sig_1 = s_1$ and $R_1$
is a product of some of the factors $z_1-\om_i(z_2,\cdots,z_n)$ including $i=1$,
and we go on with $R_2$ minimal polynomial of one of the~$\om_i^J$'s
which has not been included in~$R_1^J$, etc. \ \\
\ \\
Up to now, we have $\widetilde R \in \mathbb{F}[z_1]$ as announced, and we have
decompositions for~$Q$ in $\mathbb{F}[z_1]$:
\[
Q(z_1,\cdots,z_n) = b_M(z_2,\cdots,z_n) R_1(z_1,\cdots,z_n)^{\sig_1}\cdots R_m(z_1,\cdots,z_n)^{\sig_m}.
\]
And each factor~$R_j\ (j=1,\cdots,m)$ can be written as
%
%
%
%
%
%
%
%
\[R_j(z_1,\cdots,z_n)=\frac{1}{L_j(z_2,\cdots,z_n)} \widehat R_j(z_1,\cdots,z_n)\]
taking for~$L_j$ the l.c.m. of the denominators of the coefficients
of~$R_j$ in~$\mathbb{F}$, and $\widehat{R}_j(z_1,\cdots,z_n)$ is a primitive polynomial in $\C[z_2,\cdots,z_n][z_1]$.
Guass's lemma implies that the coefficients of $\widehat R_1^{\sig_1}\cdots \widehat R_m^{\sig_m}$ are relatively prime in $\C[z_2,\cdots,z_n]$.
Hence the coefficients of $Q$ are also in $\C[z_2,\cdots,z_n]$, which implies that $\frac{b_M}{(L_1)^{\sigma_1} \cdots (L_m)^{\sigma_{m}}} \in \C[z_2,\cdots,z_n]$. We define
\[P=\frac{b_M}{L_1^{\sig_1 -1} \cdots L_m^{\sig_m -1}}R_1\cdots R_m = \frac{b_M}{L_1^{\sig_1} \cdots L_m^{\sig_m}} \widehat R_1 \cdots \widehat R_m \]
which is the desired $z_1$-simple polynomial since $R_1 \cdots R_m$ have distinct root in $z_1$ by the construction. Finally, $P^{-1}(0)=Q^{-1}(0)$ is obviously because both $P$ and $Q$ have common factors $\frac{b_M}{L_1^{\sig_1} \cdots L_m^{\sig_m}}, \widehat R_1, \cdots, \widehat R_m$ which are all polynomials in $\C[z_2,\cdots,z_n][z_1]$.
\end{proof}
\begin{lemma}\label{lemmaz1simple}
Given a $z_1$-simple polynomial $F(z_1,\cdots,z_n)\in \C[z_1,\cdots,z_n]$, $M$ be the highest power of $z_1$, then, \ \\
$\diamond$ $G(p,z,z_2,\cdots,z_n):=F(p+z,z_2,\cdots,z_n)$ which is contained in $\C[p,z,z_2,\cdots,z_n]$ is both $p$-simple polynomial and $z$-simple polynomial.   \ \\
$\diamond$ $G(\xi,z,z_2,\cdots,z_n):=z^MF(\frac{\xi}{z},z_2,\cdots,z_n)$ which is contained in $\C[\xi,z,z_2,\cdots,z_n]$ is both $\xi$-simple polynomial and $z$-simple polynomial.
\end{lemma}

\begin{proof}
  The proof is standard and left to the reader.
\end{proof}

\begin{lemma}\label{lemmaz1Syl}
If $F(z_1,\cdots,z_n)$ is a $z_1$-simple polynomial, which means that
\[F(z_1,\cdots,z_n)=b_M(z_2,\cdots,z_n) \prod\limits_{\alpha=1}^M \big(z_1-\om_{\alpha}(z_2,\cdots,z_n) \big)\]
with $b_M\neq0$, $\om_{\alpha_1} \neq\om_{\alpha_2} $ if $1\leq \alpha_1<\alpha_2 \leq M$, then
\begin{equation}\label{formulaz1simple}
\big\{ (z_2,\cdots,z_n) \mid \om_{\alpha_1}(z_2,\cdots,z_n)=\om_{\alpha_2}(z_2,\cdots,z_n) \text{ for some } 1\leq \alpha_1<\alpha_2 \leq M
 \big\}
\end{equation}
 is an algebraic variety.
\end{lemma}

\begin{proof}
The set (\ref{formulaz1simple}) is actually $\{(z_2,\cdots,z_n) \mid Syl(F,\partial_{z_1}F)=0\}$, where $Syl(\cdot,\cdot)$ means the Sylvester matrix, with considering $F$ and $\partial_{z_1}F$ as polynomials in $z_1$ variable with coefficients in $\C[z_2,\cdots,z_n]$. So that (\ref{formulaz1simple}) is an algebraic variety generated by one polynomial since every element in the Sylvester matrix is a polynomial of $z_2,\cdots,z_n$. See \cite{GKZ94} for details.
\end{proof}

%
%


\section{Convolution Product} \label{proofconvo}

%

In this section, our goal is to prove lemma \ref{convolemma}. Let $F(z_1,\cdots,z_n)\in\widehat{\mathcal{Q}}_{\mathcal{A}}^n$, which means it is holomorphic at origin and there exists an algebraic variety $V_F\subset\mathbb{C}^n$ such that $F$ admits analytic continuation along any path which starts near origin and avoids $V_F$.
From the definition of $f(z,z_2,\cdots,z_n)$ in formula (\ref{convofor}), it is obvious that  $f(z,z_2,\cdots,z_n)\in \mathbb{C}\{z,z_2,\cdots,z_n\}$ since $F(z_1,\cdots,z_n)\in \mathbb{C}\{z_1,\cdots,z_n\}$ and $\overline{P}$ is a polynomial which vanishes at origin. The remaining part will be proved by constructing an algebraic variety $V_f$ which $f$ should avoid in general. \ \\

By the proposition \ref{propositionz1}, let us assume
\[V_F=\bigcap\limits_{J=1}^K V_F^J,
\quad
V_F^J:=\{(z_1,\cdots,z_n)\in\mathbb{C}^n\mid P_F^J(z_1,\cdots,z_n)=0\},\text{ for } J=1,\cdots,K , \]
where $P_F^J,J=1,\cdots,K$, are $z_1$-simple polynomials over $\mathbb{C}^n$. \ \\

We shall construct algebraic variety $V_f^J,\ J=1,\cdots,K$ correspondingly, s.t. if $F$ admits analytic continuation along any path which avoids the set $(P_F^J)^{-1}(0)$, then $f$ admits analytic continuation along any path which avoids the set $V_f^J$. Thus finally, the avoidant set of $f$ is an algebraic variety
\[V_f=\bigcap\limits_{J=1}^K V_f^J.\]

%


%


With a slight abuse of the notation, let
\begin{equation}\label{PJ}
 P_F^J(z_1,\cdots,z_n)=b_M(z_2,\cdots,z_n) \prod\limits_{\alpha=1}^{M} \big(z_1-\om_{\alpha}(z_2,\cdots,z_n) \big)
\end{equation}
with $b_M\neq0$, $\om_{\alpha_1} \neq\om_{\alpha_2} $ if $1\leq \alpha_1<\alpha_2 \leq M$, $\omega_{\alpha}$'s are contained in $\ov{\mathbb{F}}$.
One may keep in mind that the notations $b_M$, $M$, and $\om_{\alpha}$'s are actually depend on $J$.

\begin{definition}\label{defSG}
We define $V_f^J$ in following two cases. \ \\
\textbf{Case 1:} If
\begin{equation}\label{convassumption}
  P_F^J(\overline{P}(z,z_2,\cdots,z_n),z_2,\cdots,z_n)\neq0,   \quad  \text{for some } z,z_2,\cdots,z_n,
\end{equation}
then
\begin{equation}\label{convsimcondition}
V_f^J:=\mathbb{C}^n-
\left\{
(z,z_2,\cdots,z_n)\in\mathbb{C}^n \
\left| \
\begin{split}
  &{(P_F^J)}_{(z_2,\cdots,z_n)}(z_1) \text{ has $N$ distinct non-zero roots } and \\
  &{(P_F^J)}_{(z_2,\cdots,z_n)}(\overline{P})\neq0
\end{split}
\right.
\right\},
\end{equation}
where ${(P_F^J)}_{(z_2,\cdots,z_n)}(z_1):=P_F^J(z_1,\cdots,z_n)$ is treated as a polynomial of one variable $z_1$ with the coefficients in $\mathbb{C}[z_2,\cdots,z_n]$.  \ \\
\textbf{Case 2:} If
\begin{equation}
P_F^J(\overline{P}(z,z_2,\cdots,z_n),z_2,\cdots,z_n)=0,  \quad  \text{for all } z,z_2,\cdots,z_n,
\end{equation}
then
  \begin{equation}\label{defVfspecial}
V_f^J:=\mathbb{C}^n-
\left\{
(z,z_2,\cdots,z_n)\in\mathbb{C}^n \
\left| \
  (P_F^J)_{(z_2,\cdots,z_n)}(z_1) \text{ has $N$ distinct non-zero roots }
\right.
\right\}.
\end{equation}

\end{definition}

More precisely, the set
(\ref{convsimcondition}) is equivalent to
\begin{equation}\label{defVf}
V_f^J=\left\{(z,z_2,\cdots,z_n)\in\mathbb{C}^n \left|
 \begin{split}
 &b_M(z_2,\cdots,z_n)=0   &\text{ or } \\
 &P_F^J(0,z_2,\cdots,z_n)=0  &\text{ or }        \\
 &(P_F^J)_{(z_2,\cdots,z_n)}(z_1)  \text{ has multiple root}  &\text{ or }      \\
 &P_F^J(\overline{P},z_2,\cdots,z_n)=0  &
 \end{split}\right.
 \right\},
\end{equation}
and the set (\ref{defVfspecial}) is equivalent to
\begin{equation}\label{defVfcase2}
V_f^J=\left\{(z,z_2,\cdots,z_n)\in\mathbb{C}^n \left|
 \begin{split}
 &b_M(z_2,\cdots,z_n)=0   &\text{ or } \\
 &P_F^J(0,z_2,\cdots,z_n)=0  &\text{ or }        \\
 &(P_F^J)_{(z_2,\cdots,z_n)}(z_1)  \text{ has multiple root}  &\text{ or }
 \end{split}\right.
 \right\}.
\end{equation}

In both case, one may observe that, by the assumption (\ref{PJ}) and lemma \ref{lemmaz1Syl}, $V_f^J$ we defined above is an algebraic variety.

From the discussion above, to prove lemma \ref{convolemma}, the following lemma is needed.\ \\

\begin{lemma} \label{convlemmaJ}
  We suppose $f(z,z_2,\cdots,z_n):=\int_0^{\overline{P}(z,z_2,\cdots,z_n)} F(z_1,\cdots,z_n) dz_1 $, where $F$ holomorphic at origin, admits analytic continuation along any multi-path which avoids the algebraic variety $V_F^J=\{(z_1,\cdots,z_n)\in\mathbb{C}^n \mid P_F^J(z_1,\cdots,z_n)=0\}$. Then $f$ holomorphic at origin and it admits analytic continuation along any $\gamma$ which avoids $V_f^J$ defined above.
\end{lemma}

%
%

%

Now we only think about case 1. Case 2 will be discussed at the end of this section. We begin with a definition of $\gamma$- homotopy. \ \\

\begin{definition}
\quad For a path $\gamma(t):=(\gamma_{z}(t),\gamma_{z_2}(t),\cdots,\gamma_{z_n}(t))\in \mathbb{C}^n$, a continuous map $H:[0,1]\times[0,1]\rightarrow\mathbb{C},(t,s)\mapsto H_t(s):=H(t,s)$ is called a $\gamma$-homotopy if for any $s,t\in[0,1]$,
\begin{equation}\label{convgammahomotopy}
H_t(0)=0;\quad H_0(s)=s\cdot \overline{P}(\gamma(0));\quad H_t(1)=\overline{P}(\gamma(t));\quad P_F^J(H_t(s),\gamma_{z_2}(t),\cdots,\gamma_{z_n}(t))\neq0.
\end{equation}
\end{definition}

To prove the lemma \ref{convlemmaJ} in case 1 is sufficient to prove the following two claims.

\begin{claim}\label{germ}
Let $\gamma:[0,1]\rightarrow \mathbb{C}^n-{V}_{f}^J$ be a path such that $\gamma(0)$ near origin. If there exists a $\gamma$-homotopy, then $f$ can be analytically continued along $\gamma$.\ \\
\end{claim}

\begin{claim}\label{exist}
For any path $\gamma:[0,1]\rightarrow\mathbb{C}^n-{V}_f^J$ such that $\gamma(0)$ near origin, there exists a $\gamma$-homotopy.\ \\
\end{claim}

\begin{flushleft}
\textbf{The proof of claim \ref{germ}}\ \\
\end{flushleft}
For all $t\in [0,1]$, $\gamma|_t$ is a truncated path defined as follows:
\[\gamma|_t:[0,t]\rightarrow \mathbb{C}^n, \tau\mapsto \gamma|_t({\tau}):=\gamma(\tau).\]
We denote $cont_{\gamma|_t}f$ by the holomorphic germ at $\gamma(t)$ which is obtained by $f$ continued analytically along $\gamma$.
One may prove that, if there exists $\gamma$-homotopy, then the analytic germ at $\gamma(t)$ of $f$ is
\begin{equation}\label{convanacont}
\begin{split}
  \left(cont_{\gamma|_t} f\right)(z,z_2,\cdots,z_n)&=\int_{H_t} \left(cont_{(H|_t(s),\gamma_{z_2}|_t,\cdots,\gamma_{z_n}|_t)}F\right)(z_1,\cdots,z_n) dz_1 \\
  &+\int_{\overline{P}(\gamma(t))} ^{\overline{P}(z,z_2,\cdots,z_n)} \left(cont_{(H|_t(1),\gamma_{z_2}|_t,\cdots,\gamma_{z_n}|_t)}F\right)(z_1,\cdots,z_n) dz_1,
\end{split}
\end{equation}
where $H|_t(s)$ is the truncated path of $H_t(s)$ when we fixed $s$ in $[0,1]$. The proof of this claim is concluded with $\left(cont_{\gamma|_{t_1}} f\right)(z,z_2,\cdots,z_n) = \left(cont_{\gamma|_{t_2}} f\right)(z,z_2,\cdots,z_n)$ when $t_1$ and $t_2$ close enough by using Cauchy integral. See \cite{S14} for details.\ \\

\begin{flushleft}
\textbf{The proof of claim \ref{exist}}\ \\
\end{flushleft}

When $f$ continued analytically along $\gamma$ which starts near origin and avoids $V_f^J$, the corresponding $\gamma$-homotopy $H_t(s)$ has some moving points to be avoid, i.e. the germ of $F$ at $(H_t(s),\gamma_{z_2}(t),\cdots,\gamma_{z_n}(t))$ should be well-defined (see formula (\ref{convanacont})). From the last condition of (\ref{convgammahomotopy}) and the form of $P_F^J$ (see formula (\ref{PJ})), we know these moving points are
\begin{equation}
  \omega_{i}(t):=\omega_{i}(\gamma_{z_2}(t),\cdots,\gamma_{z_n}(t))\in\mathbb{C}, \ \text{for } i=1,\cdots,M.
\end{equation}
In the set (\ref{defVf}), the first and third conditions mean that there are always $M$ distinguished moving points $\omega_i(t)$. The second and fourth conditions mean that these $w_i(t)$'s will not touch the starting point of the homotopy $H_t(0)=0$ and the ending point of the homotopy $H_t(1)=\overline{P}(\gamma_z(t),\gamma_{z_2}(t),\cdots,\gamma_{z_n}(t))$, correspondingly.

Now we want to find the $\gamma$-homotopy $H_t(s)$.
The idea is to find a family of maps $(\Psi_t)_{t\in[0,1]}:\mathbb{C}\times\mathbb{R}\rightarrow\mathbb{C}\times\mathbb{R}$ such that for any $s$, $H_{t}(s):=\Pi_{\mathbb{C}}( \Psi_t(H_0(s)))$ yield the desired homotopy, where $\Pi_{\mathbb{C}}$ is the projection from $\mathbb{C}\times\mathbb{R}$ to $\mathbb{C}$.
Let $\omega_0(t):=\overline{P}(\gamma(t))$. If $\gamma(t)$ avoids $V_f^J$, thanks to our assumption (\ref{convassumption}), we have
\begin{equation}\label{convhomosingularcondition}
\omega_{i}(t)\neq \omega_0(t),\ \omega_i(t)\neq 0 \ \text{and} \  \omega_i(t)\neq\omega_j(t) \quad \forall i,j=1,\cdots,M,\ i\neq j \text{ and } \forall t\in[0,1].
\end{equation} See the figure 1 below.\ \\

\begin{figure}[h]\label{convpic}
\includegraphics[width=6in]{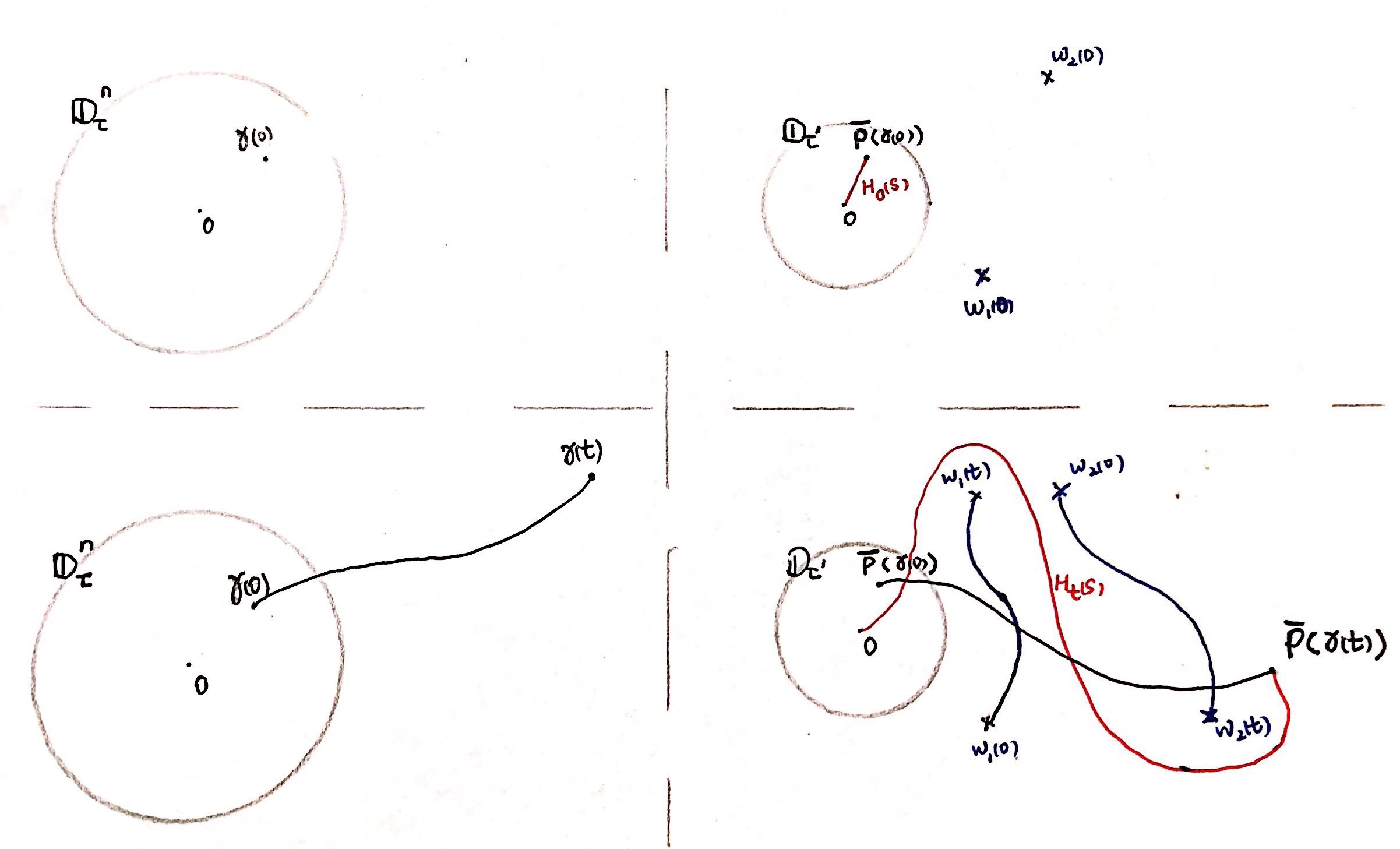}
\caption{Upper-left: there exists some $\tau$ s.t. $f$ holomorphic in $\mathbb{D}_\tau^n$ and $f$ continued analytically start at $\gamma(0)$. Upper-right: Integrate $F$ over a line segment $H_0(s)$ contained in $\mathbb{D}_{\tau^{\prime}}$. And when $f$ continued analytically along $\gamma(t)$ in the lower-left picture, the corresponding homotopy $H_t(s)$(red curve in lower-right) always exists cause the conditions (\ref{convhomosingularcondition}).}
\end{figure}
To find the $\gamma$-homotopy is sufficient to find the injective maps $\Psi_t:\mathbb{C}\times\mathbb{R}\rightarrow\mathbb{C}\times\mathbb{R}$ satisfies the following conditions:\ \\
\begin{equation}\label{conditionflow}
\begin{split}
&(1^\circ): \Psi_0=id; \\
&(2^\circ): \Psi_t(0,0)=(0,0); \\
&(3^\circ):\Psi_t(\omega_0(0),0)=(\omega_0(t),\mathfrak{L}_{\omega_0(t)}), \text{ where } \mathfrak{L}_{\omega_0(t)}:=\int_0^t|\omega_0^\prime(s)|ds. \\
&(4^\circ): \Psi_t(\omega_{i}(0),\lambda)=(\omega_{i}(t),\lambda+\mathfrak{L}_{\omega_{i}(t)}), \text{ for } i=1,\cdots,M.
\end{split}
\end{equation}

Indeed, the maps $\Psi_t$ will be generated by the flow of a non-autonomous vector field $(X(\xi,\lambda,t),|X(\xi,\lambda,t)|)$ defined as follows:
\[X(\xi,\lambda,t)
=\sum\limits_{i=0}^N
\frac{N_{i}(\xi,\lambda,t)}{N_{i}(\xi,\lambda,t)+\eta_{i}(\xi,\lambda,t)}
(\omega_{i}^{\prime}(t),|\omega_{i}^{\prime}(t)|),\]
where
\[\begin{split}
&N_0(\xi,\lambda,t):=dist((\xi,\lambda),(S(t),\mathbb{R})\cup\{(0,0)\});\\
&N_{i}(\xi,\lambda,t):=dist((\xi,\lambda),(S_{i}(t),\mathbb{R})\cup\{(0,0)\}\cup\{(\omega_0(t),\mathfrak{L}_{\omega_0(t)})\})\ \text{for}\ i=1,\cdots,N;\\
&\eta_0(\xi,\lambda,t):=dist((\xi,\lambda),(\omega_0(t),\mathfrak{L}_{\omega_0(t)}));\\
&\eta_{i}(\xi,\lambda,t):=|\xi-\omega_{i}(t)|,\ \text{for}\ i=1,\cdots,N;\\
&S(t):=\{\omega_1(t),\cdots,\omega_n(t)\};\\
&S_{i}(t):=S(t)-\{\omega_{i}(t)\}.
\end{split}
\]
One can check $N_{i}+\eta_{i}\neq0$ for $i=1,\cdots,N$.
The Cauchy-Lipschitz theorem on the existence and uniqueness of solutions to differential equations applies to $\frac{(d\xi,d\lambda)}{dt}=(X(\xi,\lambda,t),|X(\xi,\lambda,t)|)$: for every $(\xi,\lambda)\in\mathbb{C}\times\mathbb{R}$ and $t_0\in[0,1]$ there is a unique solution $t\mapsto\Psi^{t_0,t}(\xi,\lambda)$ such that $\Psi^{t_0,t_0}=id$.\ \\
Let us set $t_0=0$ and $\Psi_t:=\Psi^{0,t}$ for $t\in[0,1]$. It is easy to see that this family of maps are injective and satisfy the conditions(\ref{conditionflow}) of $\Psi$. We conclude the proof of lemma \ref{convlemmaJ} in case 1. \ \\

Here are two simple examples in case 1.

\begin{example}
$f(z,z_2):=\int_0 ^{z} \frac{1}{z_2z_1+1} dz_1 = \frac{1}{z_2} log(z_2 z+1)  $.
One can find multi-path in $\mathbb{C}^2$ to prove that the singular set of $f$ is $\{(z,z_2) \mid z_2(z_2 z + 1)=0\}$. And $\{z_2=0\}$ is actually the first condition in (\ref{defVf}), $\{z_2 z + 1=0\}$ is the fourth condition.
\end{example}

\begin{example}
  $f(z,z_2):=\int_0^z \frac{1}{(z_1+1)(z_1+z_2+1)}dz_1=\frac{1}{z_2}(log(z+1)-log(z+z_2+1)+log(z_2+1))$. The singular set of $f$ is $\{(z,z_2) \mid z_2(z+1)(z+z_2+1)(z_2+1)=0\}$. $\{z_2=0\}$ is actually the third condition in (\ref{defVf}), $\{(z+1)(z+z_2+1)=0\}$ is the fourth condition, And $\{z_2+1=0\}$ is the second condition.
\end{example}

\begin{remark}
  Although the definition $V_f^J$ gives the possibly singular set, which means that maybe a subset of $V_f^J$ is regular, from these two simple examples, one can observe that all the conditions in (\ref{defVf}) make sense.
\end{remark}
Now we discuss the lemma \ref{convlemmaJ} in case 2. The following example helps us to understand how case 2 happens.

\begin{example}
If
\[f(z,z_2):=\int_0^{z_2} \frac{1}{z_2-z_1}log(1-(z_2-z_1))dz_1,\]
we know $F(z_1,z_2):=\frac{1}{z_2-z_1}log(1-(z_2-z_1))$ is holomorphic
at $(0,0)$ and it has singular set $V_f=\{(z_1,z_2) \mid
(z_1-z_2)(z_1-(z_2-1))=0\}$. Thus $\omega_1=z_2=\overline{P}$,
$\omega_2=z_2-1$. After change the variable $u=z_2-z_1$, we have
\beglab{eqLideux}
f(z,z_2)=\int_0^{z_2} \frac1u log(1-u)du = - \Li_2(z_2).
\edla
It is obvious that the singular set of $f$ is  $\{z_2=0,1\}$ if we compute partial derivative $\frac{\partial f}{\partial z_2}$.
\end{example}

We will use the following homotopy in case 2:
\begin{definition}
\quad For a multi-path $\gamma(t):=(\gamma_{z}(t),\gamma_{z_2}(t),\cdots,\gamma_{z_n}(t))\in \mathbb{C}^n$, a continuous map $H:[0,1]\times[0,1]\rightarrow\mathbb{C},(t,s)\mapsto H_t(s):=H(t,s)$ is called a $\gamma^{\prime}$-homotopy if for any $t\in [0,1]$,
\begin{equation}\label{convgammaprimehomotopy}
\begin{split}
  &H_t(0)=0;      \quad
   H_0(s)=s\cdot \overline{P}(\gamma(0))\ \forall s\in [0,1];   \quad
   H_t(1)=\overline{P}(\gamma(t));                                \\
  &P_F^J(H_t(s),\gamma_{z_2}(t),\cdots,\gamma_{z_n}(t))\neq0\ \forall s\in [0,1).
\end{split}
\end{equation}
\end{definition}

In order to prove lemma \ref{convlemmaJ}, we will use the same procedure as in case 1.
We shall use the same formula (\ref{convanacont}) to write down the analytic continuation of $f$ along $\gamma(t)$. The only difference between $\gamma^{\prime}$-homotopy and $\gamma$-homotopy is the ending points (when $s=1$) of the fourth condition in (\ref{convgammahomotopy}) and (\ref{convgammaprimehomotopy}). Thus we shall prove the germs \[\left(cont_{(H|_t(1),\gamma_{z_2}|_t,\cdots,\gamma_{z_n}|_t)}F\right)(z_1,\cdots,z_n)\]
inside integral representation (\ref{convanacont}) are well-defined.

Let $\omega_{i}(t):=\omega_{i}(\gamma_{z_2}(t),\cdots,\gamma_{z_n}(t))\in\mathbb{C}, \ \text{for } i=1,\cdots,M \text{ and } \ov{P}(\gamma(t))=\om_1(t).$

 Actually, $\om_1(t)$ is not a singular point for $\gamma^{\prime}$-homotopy because variable ($z_1-\overline{P}$) always lies in the principle sheet when we move along $\gamma$. It will be clear after we change the variable:
 \[f(z,z_2,\cdots,z_n)=\overline{P}(z,z_2,\cdots,z_n) \int_{-1}^{0} G(\zeta,z_2,\cdots,z_n) d\zeta \]
 with $G(\zeta,z_2,\cdots,z_n):=F\left(\overline{P}(z,z_2,\cdots,z_n)(1+\zeta)
 ,z_2,\cdots,z_n\right)$. We can find sufficient small $R>0$ s.t.
 \[G(\zeta,z_2,\cdots,z_n)=\frac{1}{2\pi i} \oint_{\partial \mathbb{D}_R} \frac{G(\xi,z_2,\cdots,z_n)}{\xi-\zeta} d\xi.  \]
 Indeed, the set (\ref{defVfspecial}) which $\gamma(t)$ should avoid implies that the moving singular points of $G$
 \[\eta_i(t):= \frac{\omega_i(t)}{\omega_{1}(t)}-1,\quad \text{for}\ i=2,\cdots,M\]
 never touch $0$. This allows us to choose sufficient small $R$ s.t. %
 $\eta_i(t)$ always lie outside $\mathbb{D}_R$. %
 One can prove that such $G$ always holomorphic at $(0,\gamma_{z_2}(t),\cdots,\gamma_{z_n}(t))$. %
 We conclude the proof of lemma \ref{convlemmaJ} in case 2.\ \\

 \begin{proof}[\textbf{Proof of lemma \ref{convolemma2}}]
   Given $\overline{P}(z_2,\cdots,z_n)\in \C[z_2,\cdots,z_n]$, we may apply lemma \ref{convolemma} treating $\overline{P}$ as an element of $\C[z_1,\cdots,z_n]$: $f(z_2,\cdots,z_n)=g(z,z_2,\cdots,z_n)$ with $g\in \widehat{\mathcal{Q}}^{\mathcal{A}}_n$. We observe that formulas (\ref{defVf}) and (\ref{defVfcase2}) yield $V_g=\C\times V_f$, hence $f\in\widehat{\mathcal{Q}}^{\mathcal{A}}_{n-1}$.
 \end{proof}

\section{Hadamard Product}\label{proofHa}

\subsection{Introduction to Hadamard product on $\mathbb{C}$}\label{secHadonC}

In this section, we study the analytic continuation of the Hadamard product.

\begin{definition}
  Let $f(\xi),g(\xi)\in\mathbb{C}[[\xi]]$, $f(\xi)=\sum\limits_{n=0}^\infty a_n\xi^n$, $g(\xi)=\sum\limits_{m=0}^\infty b_m\xi^m$, we define Hadamard product:
  \[f\odot g(\xi)=\sum\limits_{n=0}^\infty a_nb_n\xi^n.\]
\end{definition}

A fact is that if $f,g\in\mathbb{C}\{\xi\}$, i.e. $f$ and $g$ have
positive radius of convergence $R_f$ and $R_g$ correspondingly, then
$f\odot g\in\mathbb{C}\{\xi\}$ and $R_{f\odot g}\leq R_fR_g$. The
following theorem is related to the classical ``Hadamard
multiplication theorem'', and is in fact a weaker version of a
theorem proved in \cite{LSS}.


\begin{theorem}
  \label{hadathm}
  If $f,g\in\widehat{\mathcal{Q}}^{\mathcal{A}}_1 $, which means, $f,g\in\C\{\xi\}$ and they admit analytic continuation along any path which avoids finite points sets $S_f$ and $S_g$ correspondingly, then $f\odot g\in\widehat{\mathcal{Q}}^{\mathcal{A}}_1$ and it admits analytic continuation along any path which starts near origin and avoids $\{0\}\cup S_f\cdot S_g$.
\end{theorem}

We may give a sketch of our proof. Choose $\xi$ inside $\mathbb{D}(R_f R_g):=\{z\in\mathbb{C} \mid |z|<R_f R_g \}$ and choose $c>0$ s.t.
\[\frac{|\xi|}{\delta_g}<c<\delta_f,\]
we have the integration representation of Hadamard product:
\beglab{eqmultconv}
f\odot g(\xi)=\frac{1}{2\pi i}\oint_C
f(z)g(\frac{\xi}{z})\frac{dz}{z},
\edla
where $C$ is a circle which can be parametrized by $ce^{is},0\leqslant s\leqslant2\pi. $\ \\

\begin{example}
  If $f(\xi)=log(1-\xi)$, then one can compute
  \[\frac{d}{d\xi} (f\odot f)(\xi)= -\frac{1}{\xi} log(1-\xi), \]
 which means the singular points of $f\odot f$ are~$0$ and~$1$.
  In fact $f\odot f = \Li_2$ as in~\eqref{eqdefLid}.
\end{example}

In order to prove the theorem \ref{hadathm}, we shall use the following homotopy:
\begin{definition}
  $\gamma^C$-homotopy for $1$ dimension case is a continuous map $H:(t,s)\in[0,1]\times[0,2\pi]\rightarrow H_t(s)\in\mathbb{C}$ such that:
  \[H_0(s)=ce^{is},\ \ H_t(s)\neq0,\ \ H_t(s)\cap S_f=\emptyset,\ \ \frac{\gamma(t)}{H_t(s)}\cap S_g=\emptyset.\]
\end{definition}

\begin{claim}\label{lemmachomotopy}
If there exists a $\gamma^C$-homotopy, then we can do analytic continuation in the following way:
\begin{equation}\label{cont}
(cont_{\gamma|_t}f\odot g)(\xi)=\frac{1}{2\pi i}\oint_{H_t}(cont_{H_{|t}(s)}f)(z)(cont_{\frac{\gamma|_t}{H_{|t}(s)}}g)({\frac{\xi}{z}})\frac {dz}z.
\end{equation}
\end{claim}

\begin{proof}
Similar to the proof of claim \ref{germ}.
\end{proof}

\begin{claim}
  If $\gamma$ starts near origin and it avoids $\{0\}\cup S_f\cdot S_g$, then there exists $\gamma^C$-homotopy.
\end{claim}

\begin{proof}
Suppose $S_f=\{f_1,\cdots,f_s\}$, $S_g=\{g_1,\cdots,g_r\}$. Without loss of generality, we assume $g_i\neq0$ for all $i=1,\cdots,r$. Indeed, if $g_1=0$, it has no influence on the $\gamma^C$-homotopy because $\gamma(t)$ never touch $0$. Let
\[\omega_i(t)=\frac{\gamma(t)}{g_i},\indent i=1,\cdots,r, \]
be the moving singular points of homotopy. By the assumption, $\gamma(t)\neq0$ implies that $\om_i(t)\neq\om_j(t)$ for $i\neq j$ and $\om_i(t)\neq 0$, $\gamma(t)\neq S_f\cdot S_g $ implies that $\om_j(t)\neq f_i$. \ \\

To find the homotopy is sufficient to find the flow $\Psi_t:\mathbb{C}\rightarrow\mathbb{C}$ which satisfies:\ \\
\[\begin{split}
&\Psi_0=id,\\
&\Psi_t(0)=0,\\
&\Psi_t(f_i)=f_i \text{\indent for }i=1,\cdots,s,\\
&\Psi_t(\omega_j(0))=\omega_j(t) \text{\indent  for } j=1,\cdots,r.
\end{split}\]
See the picture below.
\begin{figure}
\includegraphics[width=6in]{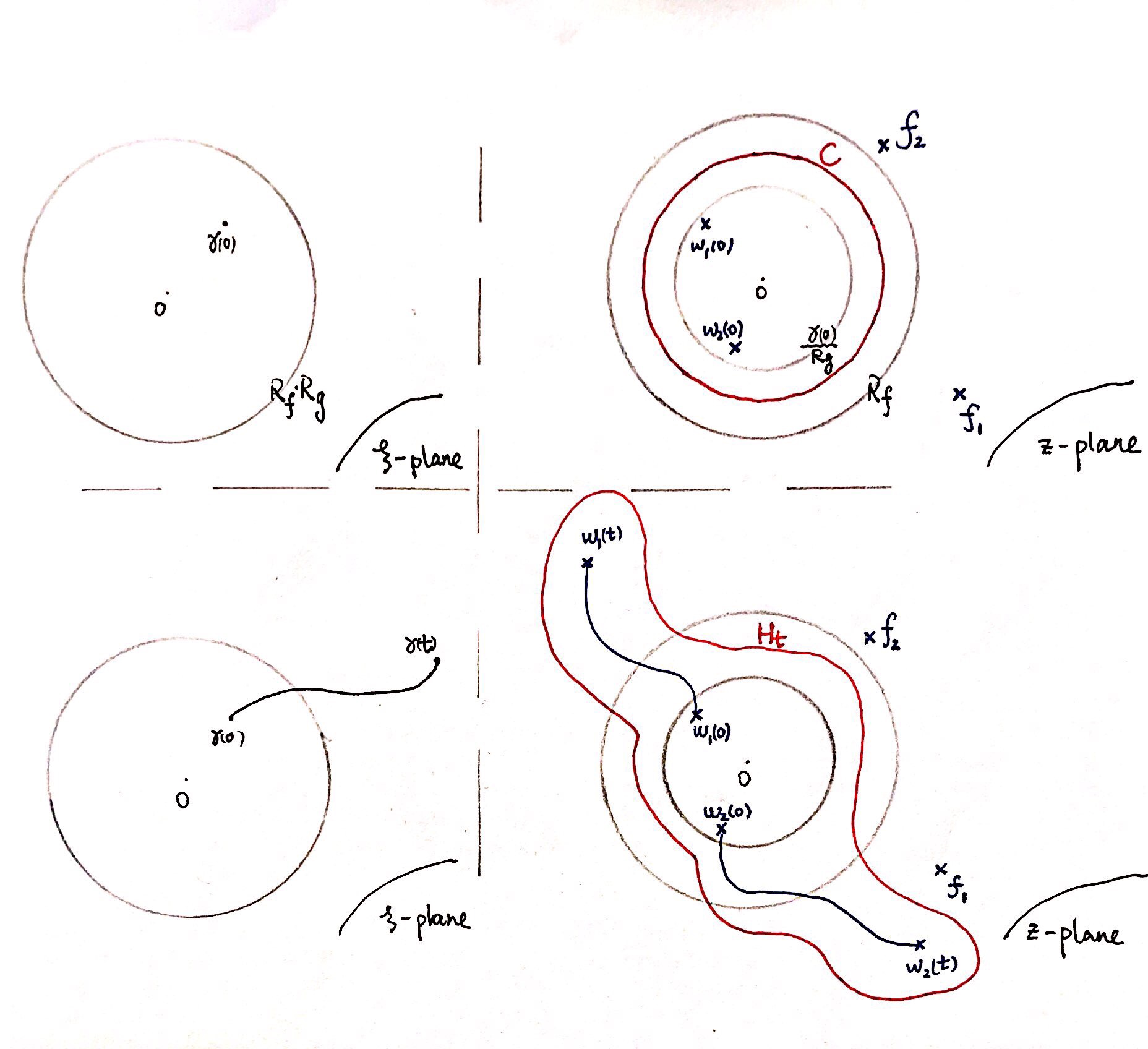}
\caption{Upper-left: $f\odot g$ is holomorphic in $\mathbb{D}_{R_f\cdot R_g}$ and $f$ continued analytically start at $\gamma(0)\in\mathbb{D}_{R_f\cdot R_g}$. Upper-right: Using $C$ to be the integral curve when $t=0$, we notice that the singular points of $f$ are outside $\mathbb{D}_{R_f\cdot R_g}$ and the ``moving singular points'' of the homotopy are inside $\mathbb{D}_{R_f\cdot R_g}$. And when $f\odot g$ continued analytically along $\gamma(t)$ in the lower-left picture, the corresponding homotopy $H_t(s)$(red curve in lower-right) always exists cause the conditions of $\omega(t)$'s.}
\end{figure}

Then we shall use the non-autonomous vector field:
\[X(\xi,t)=\sum\limits_{i=1}^r\frac{\eta_i(\xi,t)}
{\eta_i(\xi,t)+\tau_i(\xi,t)}\omega_i^\prime(t)\]
\[\text{with} \quad \eta_i(\xi,t):=dist\left(\xi,\{0\}\cup S_f\cup \bigcup\limits_{i\neq j}\omega_j(t)\right)\quad \text{and} \quad \tau_i(\xi,t):=dist(\xi,\omega_i(t)).\]

The Cauchy-Lipschitz theorem on the existence and uniqueness of solutions to differential equations applies to $\frac{d\xi}{dt}=X(\xi,t)$: for every $\xi\in\mathbb{C}$ and $t_0\in[0,1]$ there is a unique solution $t\rightarrow \Phi^{t_0,t}(\xi)$ such that $\Phi^{t_0,t_0}(\xi)=\xi$. Then we shall define our flow $\Psi_t=\Phi^{0,t}$ for $t\in[0,1]$ and the $\gamma^C$-homotopy $H_t(s):=\Phi_t(H_0(s))$.
\end{proof}

\subsection{Proof of lemma \ref{Hadalemma}}
Suppose $\hat{f}(\xi,q,p),\hat{g}(\xi,q,p)\in\widehat{\mathcal{Q}}^{\mathcal{A}}_3$, which means $\hat{f},\hat{g}$ holomorphic at $\mathbb{D}^3(\tau,0)$ and the avoidant algebraic sets are
\[\begin{split}
 V_f:=\bigcap\limits_{J=1}^R V_f^J, &\quad
 V_f^J:=\{(\xi,q,p)\in\mathbb{C}^3\mid P_f^J(\xi,q,p)=0\},\text{ for } J=1,\cdots,R , \\
 V_g:=\bigcap\limits_{K=1}^S V_g^K, &\quad
 V_g^K:=\{(\xi,q,p)\in\mathbb{C}^3\mid Q_g^K(\xi,q,p)=0\},\text{ for } K=1,\cdots,S ,
\end{split} \]
correspondingly, with $P_f^J$'s are $p$-simple polynomials and $Q_g^K$'s are $q$-simple polynomials.  \ \\

For each $J=1,\cdots,R, K=1,\cdots,S$, we shall construct algebraic variety $V_F^{JK}$, s.t. if $f$ avoids the set $\{x\in\mathbb{C}^n \mid x\in{(P_f^J)}^{-1}(0)\}$, and $g$ avoids the set $\{x\in\mathbb{C}^n \mid x\in{(Q_g^K)}^{-1}(0)\}$, then $F$ admits analytic continuation along any $\gamma$ which avoids the set $V_F^{JK}$. Thus finally, the avoidant set of $F$ is an algebraic variety
\[V_F=\bigcap\limits_{J,K} V_F^{JK}.\]

For simplifying the notation, let
\begin{equation}
\begin{split}
&P_f^J(\xi,q,p)=a_M(\xi,q) \prod\limits_{\alpha=1}^{M} \big(p-\ti{\om}_{\alpha}(\xi,q) \big), \\
&Q_g^K(\xi,q,p)=b_N(\xi,p) \prod\limits_{\beta=1}^{N} \big(q-\ti{\Om}_{\beta}(\xi,p) \big),
\end{split}
\end{equation}
with $a_M,b_N\neq0$, $\ti{\om}_{\alpha_1} \neq\ti{\om}_{\alpha_2} $ if $1\leq \alpha_1<\alpha_2 \leq M$, $\ti{\Om}_{\beta_1} \neq\ti{\Om}_{\beta_2} $ if $1\leq \beta_1<\beta_2 \leq N$, just like what we do in section \ref{proofconvo}. We shall keep in mind that the notations $a_M$, $M$, and $\ti{\om}_{\alpha}$'s are depend on $J$, and $b_N$, $N$, and $\ti{\Om}_{\beta}$'s are depend on $K$.

By the formula (\ref{Hadafor}), one can easily prove that $F$ is holomorphic inside $ \mathbb{D}(\tau)\times
\mathbb{D}(\tau)\times
\mathbb{D}(\frac{\tau^2}{4})\times
\mathbb{D}(\frac{\tau}{2})\times
\mathbb{D}(\frac{\tau}{2})$. We shall choose a point $(\xi_1,\xi_2,\xi_3,q,p)$ in this polydisc, then there exists $c>0$, s.t.
\[\frac{|\xi_3|}{\frac{\tau}{2}}<c<\frac{\tau}{2}.\]
The formula (\ref{Hadafor}) is equivalent to
\begin{equation}\label{Hadaformula}
F(\xi_1,\xi_2,\xi_3,q,p)=\frac{1}{2\pi i}\oint_{C}f(\xi_1,q,p+z)g(\xi_2,q+\frac{\xi_3}{z},p)\frac{dz}{z}
\end{equation}
where $C$ is a circle radius $c$, center at origin.\ \\
Let us consider the polynomials
\begin{equation}
\begin{split}
  &P_f^J(\xi_1,q,p+z):=a_M(\xi_1,q)\prod\limits_{\alpha=1}^M (z-\omega_{\alpha}(\xi_1,q,p))
  \in \mathbb{C}[\xi_1,q,p,z],   \\
  &Q_g^K(\xi_2,q+z,p):=b_N(\xi_2,p)\prod\limits_{\beta=1}^N (z-\Omega_{\beta}(\xi_2,q,p))
  \in \mathbb{C}[\xi_2,q,p,z].
\end{split}
\end{equation}
From the lemma \ref{lemmaz1simple}, we know that these two polynomials are both $z$-simple polynomials, which means that we have
\begin{equation}
\begin{split}
  &a_M\neq0, \quad \om_{\alpha_1} \neq \om_{\alpha_2} \text{ if } 1\leq \alpha_1<\alpha_2 \leq M,  \\
  &b_N\neq0, \quad \Om_{\beta_1} \neq \Om_{\beta_2}  \text{ if } 1\leq \beta_1<\beta_2 \leq N.
\end{split}
\end{equation}
%
We shall define the avoidant set of $F$ in $\C^5$ by using the notations above. One may notice that it is a `symmetry' condition: \ \\
\begin{definition}
  \[V_F^{JK}:= \mathbb{C}^5-
  \left\{
  (\xi_1,\xi_2,\xi_3,q,p) \in \mathbb{C}^5
   \left|
  \begin{split}
    &\mathfrak{P}_{\xi_1,q,p}(z) \text{ has } M \text{ distinct non-zero roots,\ and}\\
    &\mathfrak{Q}_{\xi_2,q,p}(z) \text{ has } N \text{ distinct non-zero roots,\ and}\\
    &\xi_3 \notin\{\om_{\alpha}\Om_{\beta}\} \cup \{0\}
  \end{split}
  \right.
  \right\},\]
where $\mathfrak{P}_{\xi_1,q,p}(z):=P_f^J(\xi_1,q,p+z)$ is treated as a polynomial of one variable $z$ with the coefficients in $\mathbb{C}[\xi_1,q,p]$ and $\mathfrak{Q}_{\xi_2,q,p}(z):=Q_g^K(\xi_2,q+z,p)$ is treated as a polynomial of one variable $z$ with the coefficients in $\mathbb{C}[\xi_2,q,p]$.
\end{definition}

\begin{remark}
In fact, $V_{F}^{JK}$ will be simplified to one sentence:
\begin{equation}\label{Hadasimcondition}
V_F^{JK}=\mathbb{C}^5-\left\{(\xi_1,\xi_2,\xi_3,q,p)\in\mathbb{C}^5 \left| z^N \mathfrak{P}_{\xi_1,q,p}(z)\mathfrak{Q}_{\xi_2,q,p}(\frac{\xi_3}{z})
\text{ has $M+N$ distinct non-zero roots} \right. \right\}.
\end{equation}
By the lemma \ref{lemmaz1Syl}, we know $V_F^{JK}$ is an algebraic variety.
\end{remark}

From the discussion above, to prove the lemma \ref{Hadalemma} is sufficient to prove the following claim:

\begin{claim}\label{HadaclaimJK}
  If $f$ and $g$ admit analytic continuation along any path which avoid $V_f^J$ and $V_g^K$ respectively, then $F$ defined by formula (\ref{Hadaformula}) admits analytic continuation along any path $\gamma$ which avoids $V_F^{JK}$ defined above.
\end{claim}

By using following seven conditions, we explain $V_F^{JK}$ more precisely,
  \begin{equation}\label{defVfog}
V_F^{JK}=\left\{(\xi_1,\xi_2,\xi_3,q,p)\in\mathbb{C}^5 \left|
 \begin{split}
 &a_M(\xi_1,q)=0   &\text{ or }& \\
 &b_N(\xi_2,p)=0   &\text{ or }&        \\
 &\om_{\alpha_1}(\xi_1,q,p)=\om_{\alpha_2}(\xi_1,q,p) \text{ for }  \alpha_1\neq \alpha_2  &\text{ or }&      \\
 &\Om_{\beta_1}(\xi_2,q,p)=\Om_{\beta_2}(\xi_2,q,p)  \text{ for }  \beta_1\neq \beta_2  &\text{ or }&      \\
 &\omega_{\alpha}(\xi_1,q,p)=\frac{\xi_3}{\Omega_{\beta}(\xi_2,q,p)} &\text{ or }&  \\
 &\omega_{\alpha}(\xi_1,q,p)=0   &\text{ or }&  \\
 &\xi_3=0
 \end{split}\right.
 \right\}
\end{equation}
where $\alpha,\alpha_1,\alpha_2=1,\cdots,M$ and $\beta,\beta_1,\beta_2=1,\cdots,N$.

We assume $\gamma:[0,1]\rightarrow \mathbb{C}^5$ starting near $0\in\mathbb{C}^5$, denoted by $\gamma(t)=(\gamma_{\xi_1}(t),\gamma_{\xi_2}(t),\gamma_{\xi_3}(t),\gamma_q(t),\gamma_{p}(t))$. We shall use the following homotopy to do analytic continuation:

\begin{definition}
  With a little abuse of name, a $\gamma^C$-homotopy for high dimension is a continuous map $H:(t,s)\in[0,1]\times[0,2\pi]\rightarrow H_t(s)\in\mathbb{C}$ s.t. for any $t,s,\alpha,\beta$:
  \[H_0(s)=ce^{is},\quad H_t(s)\neq 0,\quad H_t(s)\neq \omega_{\alpha}(\xi_1,q,p),\quad H_t(s)\neq \frac{\xi_3}{\Omega_{\beta}(\xi_2,q,p)}.\]
\end{definition}

One shall prove that if there exists such $\gamma^C$-homotopy, then
\[\begin{split}
(cont_{\gamma|_t}F)(\xi_1,\xi_2,\xi_3,q,p)
=\frac{1}{2\pi i} \oint_{H_t} \left(cont_{(\gamma_{\xi}|_t,\gamma_{q}|_t,\gamma_{p}|_t+H|_t(s))}\hat{f}\right)
(\xi_1,q,p+z)   \\
\cdot \left(cont_{(\gamma_{\xi_2}|_t,\gamma_{q}|_t+\frac{\gamma_{\xi_3}|_t}{H|_t(s)},
\gamma_{p}|_t)}\hat{g}\right)
(\xi_2,q+\frac{\xi_3}{z},p) \frac{dz}{z},
\end{split}\]
which means $F$ admits analytic continuation along $\gamma$ (see the proof of claim \ref{germ} for details).\ \\

Now we assume $\gamma$ avoids $V_F^{JK}$, to find the homotopy is sufficient to find the flow $\Psi_t:\mathbb{C}\rightarrow\mathbb{C}$ which satisfies:\ \\
\[\begin{split}
&\Psi_0=id,\\
&\Psi_t(0)=0,\\
&\Psi_t(\omega_i)=\omega_i(t) \text{\indent for }i=1,\cdots,N+M,
\end{split}\]
where $\omega_{N+j}(t):=\frac{\xi_3}{\Omega_{j}(t)}$, $j=1,\cdots,M$. Here we use $\omega_{\alpha}(t)$ and $\Omega_{\beta}(t)$ to simplify the notation $\omega_{\alpha}(\gamma_{\xi_1}(t),\gamma_{q}(t),\gamma_{p}(t))$ and $\Omega_{\beta}(\gamma_{\xi_2}(t),\gamma_{q}(t),\gamma_{p}(t))$ respectively. \ \\

From the third condition to last condition in definition (\ref{defVfog}), we know that if $\gamma$ avoids $V_F^{JK}$, then $\omega_i(t)\neq \omega_j(t)$ and $\omega_i(t)\neq0$ for $1\leq i\neq j \leq M+N$. The first two conditions in definition (\ref{defVfog}) provide that no moving singular points goes to infinity.

Thus we could use the non-autonomous vector field:
\[X(\xi,t)=\sum\limits_{i=1}^r\frac{\eta_i(\xi,t)}
{\eta_i(\xi,t)+\tau_i(\xi,t)}\omega_i^\prime(t)\]
where $\eta_i(\xi,t):=dist(\xi,\{0\}\cup \bigcup\limits_{j\neq i}\omega_j(t))$ and $\tau_i(\xi,t):=dist(\xi,\omega_i(t))$. The Cauchy-Lipschitz theorem on the existence and uniqueness of solutions to differential equations applies to $\frac{d\xi}{dt}=X(\xi,t)$: for every $\xi\in\mathbb{C}$ and $t_0\in[0,1]$ there is a unique solution $t\rightarrow \Phi^{t_0,t}(\xi)$ such that $\Phi^{t_0,t_0}(\xi)=\xi$. Then we shall define our flow $\Psi_t=\Phi^{0,t}$ for $t\in[0,1]$ and the desired $\gamma^C$-homotopy $ H_t(s):=\Psi_t(ce^{is})$. We conclude the proof of claim \ref{HadaclaimJK}.\ \\

\begin{example}
  Let $f(\xi,q,p)=log(3-\xi-q-p)$ and $g(\xi,q,p)=\frac{1}{3-\xi-q-p}\cdot\frac{1}{4-\xi-2q-p}$,
  thus $f$ and $g$ holomorphic in $\mathbb{D}(1)\times\mathbb{D}(1)\times\mathbb{D}(1)$.
  From the discussion above, we know $F$ holomorphic in $\mathbb{D}(1)\times\mathbb{D}(1)\times\mathbb{D}(\frac14)\times\mathbb{D}(\frac12)\times\mathbb{D}(\frac12)$.\ \\
  We choose $\gamma_{\xi_3}(0)=\frac18$, then there exists $c>0$ (we may choose $c=\frac38$) s.t. $\frac{\frac18}{\frac12}<c<\frac12$. From the formula (\ref{Hadaformula}), we compute
  \[\begin{split}
  F(\xi_1,\xi_2,\xi_3,q,p)
  &=\frac1{2\pi i}
  \oint_{C}
  log(3-\xi_1-q-p-z)
  \cdot
  \frac{1}{3-\xi_2-q-\frac{\xi_3}{z}-p}
  \cdot
  \frac{1}{4-\xi_2-2q-\frac{2\xi_3}{z}-p}\frac{dz}{z}   \\
  &=\frac{1}{A} \frac1B \frac1{2\pi i}
  \oint_{C}
  \frac1{z-\frac{\xi_3}{A}}
  \cdot
  \frac{1}{z-\frac{2\xi_3}{B}}
  \cdot
  zdz  \\
  &=\left[
  log(3-\xi_1-q-p-\frac{\xi_3}{A})
  \cdot
  \frac{\xi_3}{A}
  -log(3-\xi_1-q-p-\frac{2\xi_3}{B})
  \cdot
  \frac{2\xi_3}{B}
  \right]
  \cdot\frac{1}{\xi_3(B-2A)}  \\
  &=\frac{1}{\xi_2+p-2}
  \cdot
  \Big[
  log(3-\xi_1-q-p-\frac{\xi_3}{A})
  \cdot
  \frac{1}{A}
  -log(3-\xi_1-q-p-\frac{2\xi_3}{B})
  \cdot
  \frac{2}{B}
  \Big] ,
  \end{split}
  \]
where $A:=3-\xi_2-q-p$ and $B:=4-\xi_2-2q-p$. Here we know $P_f(\xi_1,q,p)=1$, $P_g(\xi_2,\xi_3,q,p)=AB$, $\omega_1=3-\xi_1-q-p$, $\Omega_1=\frac{\xi_3}{A}$, $\Omega_{2}=\frac{2\xi_3}{B}$.
Thus we can easily observe that $F$ has singular set when $A=0,B=0,\omega_1=\Omega_1, \omega_1=\Omega_2,\Omega_1=\Omega_2$. In fact, for the case $\Omega_1=\Omega_2$, which means $\xi_2+p-2=0$, is not singular set at principle sheet of germ $F$, but it may be singular set in some other sheet. \ \\
\end{example}

\begin{remark}
  The above computation gives an example for the second, fourth and fifth condition in definition (\ref{defVfog}). One may find examples for other conditions.
\end{remark}

\section{Conclusion of the proof for an arbitrary number of degrees of fredom} \label{secHigher}

Now we assume $q=(q_1,\cdots,q_r)$ and $p=(p_1,\cdots,p_r)$.
Recall the integral formula~\eqref{formulahigh}:
\begin{equation*}
  \begin{split}
  \hat{f}\ast \hat{g} (\xi,q,p) =&
  \frac{d^{r+2}}{d \xi ^{r+2}} \int^\xi_0 d\xi_1
  \int^{\xi-\xi_1}_0 d\xi_2
  \cdots \int_0^{\xi-\xi_1-\cdots -\xi_{r+1}} d\xi_{r+2}
  \big( \frac{1}{2\pi} \big)^r \overbrace{\int_0^{2\pi} \cdots \int_0^{2\pi}}^r d\theta_1 \cdots d\theta_r  \\
  &\hat{f}(\xi_{r+1}, q_1,\cdots, q_r,p_1+\sqrt{\xi_1}e^{-i\theta_1}, \cdots, p_r+\sqrt{\xi_r} e^{-i\theta_r}) \\
  & \hat{g}(\xi_{r+2}, q_1+\sqrt{\xi_1}e^{i\theta_1}, \cdots, q_r+\sqrt{\xi_r} e^{i\theta_r}, p_1, \cdots ,p_r).
\end{split}
\end{equation*}

We want to prove Theorem~\ref{thmStdN}, \ie that $\hat{f}, \hat{g} \in
\mathcal{Q}^{\mathcal{A}}_{2r+1} \imp \hat{f}\ast \hat{g} \in \mathcal{Q}^{\mathcal{A}}_{2r+1}$.

\medskip

Then, by using our lemmas on convolution product and
Hadamard product, we will be able to deduce the same result for $*_M$,
\ie Theorem~\ref{thmMoyalN}, because
\begla
\hat{f}\in\widehat{\mathcal{Q}}^{\mathcal{A}}_{2r+1} \imp
\wh T\hat{f}\in\widehat{\mathcal{Q}}^{\mathcal{A}}_{2r+1}
\edla
by the same arguments as below.

\medskip

The following definition is useful.

\begin{definition}
For $n\geq2$, $i\leq j\leq n$,
  \[ F(z_1,\cdots,z_n) \in \C\{z_1,\cdots,z_n \} \mapsto (\odot_{ij}F)(\xi,z_1,\cdots,z_n) \in \C\{\xi,z_1,\cdots,z_n\},\]
  where
  \[(\odot_{ij}F)(\xi,z_1,\cdots,z_n):=\int_0^{2\pi} F(z_1,\cdots,z_{i-1}, z_i+\sqrt{\xi}e^{-i\theta},z_{i+1},\cdots,z_{j-1}, z_j+\sqrt{\xi}e^{i\theta} ,z_{j+1},\cdots, z_n ) \frac{d\theta}{2\pi}.\]
\end{definition}

By using this definition, we may understand the formula (\ref{formulahigh}) as
\begin{equation}
\hat{f}\ast \hat{g} (\xi,q,p) =
  \frac{d^{r+2}}{d \xi ^{r+2}} \int^\xi_0 d\xi_1
  \int^{\xi-\xi_1}_0 d\xi_2
  \cdots \int_0^{\xi-\xi_1-\cdots -\xi_{r+1}} d\xi_{r+2} G(\xi_1,\cdots,\xi_{r+2},q,p),
\end{equation}
where $G$ is obtained by $\odot^r_{2r+1,3r+2} (F)$,
\[
F(\xi_1,z_1,\cdots,z_{2r},\xi_2,z_{2r+1},\cdots,z_{4r}):= \hat{f}(\xi_1,z_1,\cdots,z_{2r})\hat{g}(\xi_2,z_{2r+1},\cdots,z_{4r})
\]
and let $z_i=z_{2r+i}=q_i$ and $z_{r+j}=z_{3r+j}=p_j$ for $i,j=1,\cdots,r$. \ \\

By lemma \ref{convolemma} and lemma \ref{convolemma2}, to prove $\hat{f}\ast \hat{g}\in \mathcal{Q}^{\mathcal{A}}_{2r+1}$ is sufficient to prove the following lemma:

\begin{lemma}
  If $F(z_1,\cdots,z_n)\in \mathcal{Q}^{\mathcal{A}}_{n}$ then $(\odot_{ij} F)(\xi,z_1,\cdots,z_n) \in \mathcal{Q}^{\mathcal{A}}_{n+1}$.
\end{lemma}

\begin{proof}
  We assume $\{(z_1,\cdots,z_n) \mid P_F(z_1,\cdots,z_n)=0\}$ is the avoidant set of $F$. Consider $Q(z,\xi,z_1,\cdots,z_n):=z^N P_F(z_1,\cdots,z_i+{z},\cdots,z_j+\frac{\xi}{z},\cdots,z_n )$, where $N$ is the smallest number such that $Q$ is a polynomial in the variables $z,\xi,z_1,\cdots,z_n$. We treat $Q_{\xi,z_1,\cdots,z_n}(z):=Q(z,\xi,z_1,\cdots,z_n)$ as the polynomial in one variable $z$, order $M$, with coefficient in $\C[\xi,z_1,\cdots,z_n]$. One may prove that, $\odot_{ij} F$ admits analytic continuation along any path which contained in the following set
\begin{equation}
\left\{(\xi,z_1,\cdots,z_n)\in\mathbb{C}^{n+1} \left| Q_{\xi,z_1,\cdots,z_n}(z)
\text{ has $M$ distinct non-zero roots} \right. \right\}.
\end{equation}
\end{proof}



\vspace{1cm}

\subsubsection*{Acknowledgements}

The first and the third authors aknowledge support from NSFC (No.11771303).

\smallskip

The second author thanks Capital Normal University for their hospitality
during the period September 2019--February 2020.

\smallskip

The third author is partially supported by NSFC (No.s 11771303, 11911530092, 11871045).

\smallskip

This paper is partly a result of the ERC-SyG project, Recursive and
Exact New Quantum Theory (ReNewQuantum) which received funding from
the European Research Council (ERC) under the European Union's Horizon
2020 research and innovation programme under grant agreement No
810573.

\vspace{1cm}



\end{document}